\documentclass[journal]{IEEEtran}
\usepackage[T1]{fontenc}
\usepackage{cite}
\ifCLASSINFOpdf
  \usepackage[pdftex]{graphicx}
\else
  \usepackage[dvips]{graphicx}
\fi
\usepackage{amssymb}
\usepackage{amsthm}
\newtheorem{remark}{Remark}

\usepackage{algorithm}
\usepackage{algorithmic}
\usepackage{tabularx}
\usepackage{array} 
\usepackage{amsthm,amssymb,amsfonts}
\usepackage{epsfig}
\usepackage{subfigure}
\usepackage{array}
\usepackage{changes}
\usepackage{textcomp}
\usepackage{xcolor}
\usepackage{amsmath}
\usepackage{fancyhdr}

\usepackage{dsfont}
\usepackage{amsmath,amssymb,amsfonts}
\usepackage{subfigure}
\usepackage{textcomp}
\usepackage{xcolor}
\usepackage{multirow}
\usepackage{mathtools}
\usepackage{mhchem}
\usepackage{tensor}
\usepackage{graphicx}  
\usepackage{psfrag}    
\usepackage{amsthm}
\usepackage{mathrsfs}
\usepackage{epstopdf}
\usepackage{bm}
\usepackage{cuted}
\usepackage{color, soul, framed}
\usepackage{amsthm}

\newtheorem{proposition}{Proposition}
\theoremstyle{remark}
\usepackage{xcolor}
\usepackage{xpatch}
\usepackage{tabularx}
\ifCLASSOPTIONcompsoc
  \usepackage[caption=false,font=normalsize,labelfont=sf,textfont=sf]{subfig}
\else
  \usepackage[caption=false,font=footnotesize]{subfig}
\fi
\usepackage{fixltx2e}
\usepackage{dblfloatfix}

\usepackage{multirow}
\usepackage{booktabs}
\usepackage{url}
\makeatletter
\def\changeBibColor#1{%
	
	\in@{#1}{}
	
	\ifin@\color{blue}\else\normalcolor\fi
	
}

\xpatchcmd\@bibitem
{\item}
{\changeBibColor{#1}\item}
{}{\fail}
\xpatchcmd\@lbibitem
{\item}
{\changeBibColor{#2}\item}
{}{\fail}
\makeatother

\begin{document}

\title{Joint 3D Trajectory and Power Allocation for HAPs–UAV Bistatic ISARAC in Low-Altitude Networks
}

\author{Bang~Huang,~\IEEEmembership{Member,~IEEE,}	
	Mohamed-Slim Alouini,~\IEEEmembership{Fellow,~IEEE,}
	\thanks{ The authors are with the Computer, Electrical and Mathematical Science and Engineering (CEMSE) division in King Abdullah University of Science and Technology (KAUST), Thuwal 6900, Makkah Province, Saudi Arabia.  (Emails: bang.huang@kaust.edu.sa; slim.alouini@kaust.edu.sa;) (Corresponding author: Bang Huang)}}
\maketitle

\begin{abstract}
This paper investigates joint three-dimensional (3D) trajectory planning and 
resource allocation for a high-altitude platform (HAPs)–unmanned aerial vehicle 
(UAV) bistatic integrated synthetic aperture radar (SAR) and communication 
(ISARAC) system in low-altitude networks. In the proposed architecture, the HAPs 
provides persistent wide-area connectivity by transmitting ISARAC waveforms for 
ground-user communications, while a low-altitude UAV exploits its proximity and 
mobility to passively collect ground-target echoes for high-resolution SAR imaging.
We formulate a sum-rate maximization problem for ground users subject to stringent 
SAR imaging signal-to-noise ratio (SNR) and resolution requirements, a total 
energy budget for ISARAC transmission, and UAV dynamic constraints. The resulting 
problem is inherently nonconvex. To tackle it, an alternating optimization (AO) 
framework is developed, where the power-allocation subproblem with fixed UAV 
states admits a closed-form water-filling solution, while the UAV trajectory 
optimization with fixed transmit powers is handled via successive convex 
approximation (SCA) and difference-of-convex (DC) programming.
Simulation results verify the effectiveness of the proposed approach and 
demonstrate its capability to jointly support persistent communication coverage 
and high-resolution sensing in low-altitude network scenarios.

\end{abstract}


\IEEEpeerreviewmaketitle
\section{Introduction}

\IEEEPARstart{B}{ridging} the digital divide has become a central vision for beyond-5G (B5G) and 6G networks, driving a paradigm shift from town-centric infrastructure toward human-centric connectivity that ensures universal access regardless of location \cite{belmekkiCellularNetworkSky2024}. Non-terrestrial networks (NTNs) \cite{rago2025innovative} are widely recognized as a key enabler of this vision, offering cost-effective and flexible coverage beyond the reach of terrestrial base stations. Among various NTN platforms, high-altitude platforms (HAPs) \cite{abbasi2024haps} and unmanned aerial vehicles (UAVs) \cite{Zeng2016Wirelesscommunications} have attracted growing attention due to their complementary characteristics, with HAPs providing wide-area and long-endurance coverage, while UAVs offering agile mobility and rapid deployment.

Beyond connectivity, 6G is envisioned to natively integrate environmental sensing, 
giving rise to the paradigm of integrated sensing and communications (ISAC) \cite{WuJiang2024AIEnhancedIntegrated}. 
Most existing ISAC studies, however, primarily focus on conventional radar 
functionalities, such as beamforming and target parameter estimation, which 
provide limited and often abstract sensing outputs.
In contrast, synthetic aperture radar (SAR) 
\cite{JirousekPeichl2024DLR,jirousek2023design} is capable of delivering 
interpretable, image-based situational awareness, offering substantially greater 
practical value from the user perspective. SAR imagery enables a direct and 
comprehensive understanding of the environment, and, unlike optical sensors, SAR 
operates robustly under all-weather and day-and-night conditions 
\cite{DengYu2024TheHighResolutionSynthetic}.
Motivated by these advantages, the integration of SAR into ISAC systems, referred 
to as integrated SAR and communication (ISARAC), has recently emerged as a 
promising research direction, particularly for mission-critical applications 
where intuitive and reliable environmental perception is essential \cite{lahmeri2022trajectory,huang2025joint}.

Despite these advances, most existing ISARAC studies remain confined to {single-platform} designs, where sensing and communication functions are tightly coupled to the same platform geometry, altitude, and energy budget. In contrast, bistatic or multistatic ISARAC architectures introduce geometric diversity and complementary viewpoints, which fundamentally expand the feasible sensing--communication tradeoff region. By decoupling sensing and communication roles across heterogeneous platforms, bistatic ISARAC systems can enable joint functionalities that are structurally difficult to realize using a single platform alone.

\subsection{Related works}
Extensive studies have examined the use of UAVs and HAPs in wireless communications 
as key enablers of NTNs. For UAVs, existing research can 
be broadly categorized into two paradigms. The first paradigm is UAV-assisted 
communications, where UAVs act as aerial base stations or relays to support 
terrestrial networks, enabling rapid connectivity restoration and traffic 
offloading \cite{Mozaffari2016UnmannedAerialVehicle,Lyu2017PlacementOptimization}. 
The second paradigm is cellular-enabled UAV communications, where UAVs operate as 
aerial users establishing bidirectional links with terrestrial networks to satisfy 
their own communication demands \cite{wu2019fundamental,wu2018capacity}. These 
advances have highlighted UAVs as a key technological pillar of the emerging 
low-altitude economy (LAE) \cite{jiang2025integrated}.

In comparison, HAPs are less maneuverable but provide much broader coverage and 
significantly longer endurance. Prior works have investigated integrated 
satellite--air--ground architectures \cite{Wu2025MultiHAP,Ali2024PowerLEO}, secure 
communication schemes, and multiple-access techniques such as non-orthogonal 
multiple access (NOMA) tailored for HAPs-based systems \cite{Javed2025SystemDesign}, 
confirming the critical role of HAPs in future NTN deployments 
\cite{Nabi2025JointOffloading,Yu2024MECNOMA}. From a resource management perspective, 
dynamic resource allocation and scheduling in HAPs-assisted networks have also 
received growing attention. For example, Kawamoto \emph{et al.} 
\cite{kawamoto2023traffic} proposed a traffic-prediction-based dynamic resource control 
strategy for HAPS-mounted mobile edge computing (MEC)-assisted satellite communication systems, while 
Dahrouj \emph{et al.} \cite{dahrouj2023machine} investigated machine 
learning-based user scheduling in integrated satellite--HAPs--ground networks, 
demonstrating the potential of data-driven approaches for large-scale NTN resource 
management.

Recent studies have increasingly explored ISAC on aerial platforms, 
particularly UAVs and HAPs, through joint design of mobility, power allocation, 
and beamforming to balance communication and sensing objectives 
\cite{WuJiang2024AIEnhancedIntegrated}.
 On the 
UAV side, extensive studies have addressed the joint design of trajectory, power 
allocation, and beamforming to balance communication and sensing requirements 
\cite{lyu2022joint,liu2024uav,liu2024secure,wang2025joint,mao2024uav}. Security-aware 
ISAC designs \cite{liu2024secure}, multi-antenna beamforming schemes 
\cite{wang2025joint,mao2024uav}, cooperative multi-UAV systems 
\cite{zhang2024joint,bayessa2024joint}, and learning-enabled UAV ISAC frameworks 
\cite{Yan2025uav} have further demonstrated the flexibility and task-adaptive 
capability of UAV-based ISAC systems.

Complementing these efforts, HAPs-based ISAC research has also gained momentum by 
leveraging the wide-area coverage and long endurance of HAPs. Representative works 
include HAPs-assisted ISAC frameworks with cooperative UAV jamming 
\cite{benaya2025aerial}, ISAC-assisted channel estimation and beam alignment for 
inter-HAPs networking \cite{kirik2025isac}, and advanced beamforming designs that 
jointly optimize sensing accuracy and communication signal to interference plus noise ratio (SINR) under stratospheric 
conditions \cite{kanani2025haps}. From a communication modeling perspective, 
practical issues such as interference management, atmospheric attenuation, and 
channel uncertainty in HAP- and UAV-assisted networks have also been extensively 
investigated 
\cite{kawamoto2023hapsbased,li2024uav,abbasi2024haps,kawamoto2024interference,
rodrigues2025weather,zedini2024improving}, providing valuable insights into realistic 
channel modeling and interference suppression in space--air--ground integrated 
networks.

More recently, ISARAC has attracted growing attention, where SAR capabilities are integrated into ISAC frameworks to enable image-based situational awareness under all-weather and day-and-night conditions. From both UAV and HAPs 
perspectives, recent studies have investigated real-time SAR imaging and trajectory 
optimization under energy constraints 
\cite{lahmeri2022trajectory,huang2025joint}. Stefano \emph{et al.} 
\cite{moro2024exploring} explored ISAC-oriented signal processing in a UAV--SAR 
setting, while Huang \emph{et al.} \cite{huang2025design} and Zhang \emph{et al.} 
\cite{zhang2025design} studied waveform and joint beamforming--trajectory design 
for HAPs-enabled ISARAC systems. From the UAV perspective, trajectory planning and 
energy-aware designs for ISARAC have also been investigated 
\cite{hu2022trajectory,Zhou2025JointTrajectory,zhou2025energy}.

Despite the extensive efforts reviewed above, existing studies predominantly 
focus on single-platform ISAC or ISARAC designs, while systematic investigations 
of bistatic or multistatic ISARAC architectures remain scarce. In particular, 
the joint system modeling and optimization of heterogeneous HAPs--UAV platforms 
for integrated communication and SAR-based sensing have not been fully explored 
in the current literature. Moreover, this paradigm naturally aligns with the evolution of 6G NTNs and space–air–ground integrated networks (SAGIN), underscoring its research significance and practical value. 

\subsection{Contribution}

Motivated by this gap, this paper investigates a HAPs–UAV cooperative bistatic ISARAC system in low-altitude networks, where HAPs provide persistent wide-area connectivity and UAVs operate at low altitude to enable high-resolution SAR sensing, through system modeling and joint optimization. These advantages make bistatic ISARAC particularly well-suited for mission-critical scenarios such as disaster emergency response. In major natural disasters including floods or earthquakes, terrestrial communication infrastructure is often damaged and unable to promptly serve affected users. In such cases, a HAPs can loiter over the disaster zone on a stable circular trajectory, providing continuous coverage and reliable backhaul via directional or multibeam transmission. Meanwhile, the rescue command center requires spatiotemporal situational awareness including population distribution, transportation blockages, and secondary hazard risks, to guide personnel deployment and resource allocation. To this end, we envision a HAPs–UAV bistatic ISARAC concept, where the HAPs transmits ISARAC waveforms and a UAV, exploiting its mobility and low-altitude proximity, passively receives echoes to perform SAR imaging, change detection, and target recognition. The UAV then forwards results to the command center, closing the decision loop. This architecture ensures both communication continuity and high-resolution sensing, offering strong survivability, rapid deployment, and persistent coverage with clear practical significance.

Based on the above envision, the main contributions of this work are summarized as follows:
\begin{itemize}
    \item To the best of the authors’ knowledge, this is the first exploration of a HAPs–UAV cooperative bistatic ISARAC system. Specifically, we establish a signal model for the coupled communication--sensing chain of the considered HAPs--UAV ISARAC system. 
In this chain, the HAPs transmits ISARAC waveforms, ground users access the downlink communication service, and the UAV passively receives ground-target echoes for SAR imaging. 
This model provides an analytically tractable physical and geometric foundation for subsequent joint optimization.
    \item We develop an optimization framework that aims to maximize the sum rate of ground users while explicitly incorporating key engineering constraints, including SAR imaging signal-to-noise ratio (SNR) and resolution requirements, total HAPs energy budget for ISARAC tasks, and UAV dynamics/kinematics. This framework quantitatively reveals the strict lower-bound effect imposed by imaging thresholds on communication degrees of freedom, as well as the inherent structural trade-off between communication and sensing.
    \item To address the resulting non-convex problem, we propose a low-complexity alternating optimization (AO) framework. When the UAV trajectory/velocity is fixed, the power allocation subproblem admits a closed-form water-filling solution. Next when the transmit power is fixed, the UAV's trajectory–velocity subproblem is convexified and reconstructed via successive convex approximation (SCA) and
difference-of-convex (DC) decomposition, enabling efficient convex optimization.
\item Simulation results first verify the convergence of the proposed algorithm. 
Moreover, sensitivity analyses under multiple parameter settings characterize how system performance varies with UAV receive SNR, HAPs transmit power, and sensing slot/dwell time. 
The results show that, under limited energy and stringent imaging thresholds, the proposed framework achieves more favorable energy allocation and UAV trajectory--geometry configuration, offering practical guidelines for parameter selection in real deployments. 
In addition, comparison with a single-platform HAPs-based ISARAC benchmark demonstrates that the proposed HAPs--UAV bistatic architecture provides a wider feasible operating region by reducing the sensing-induced energy floor under stringent SAR SNR constraints.
\item The proposed model and algorithm directly match the dual requirements of persistent coverage and high-resolution sensing in scenarios such as disaster emergency response. HAPs provide steady coverage and backhaul, while UAVs, leveraging their proximity and maneuverability, perform SAR imaging and change detection.
\end{itemize}
\subsection{Outline}
The remainder of this paper is organized as follows. Section~\ref{sec2} presents the system model of the bistatic HAPs--UAV ISARAC framework. Section~\ref{sec3} formulates and solves the joint optimization problem. Simulation results are provided in Section~\ref{sec4}, followed by conclusions in Section~\ref{sec5}.

\section{System model}
\label{sec2}
\begin{figure}[htp]
	\centering
	{\includegraphics[width=0.35\textwidth]{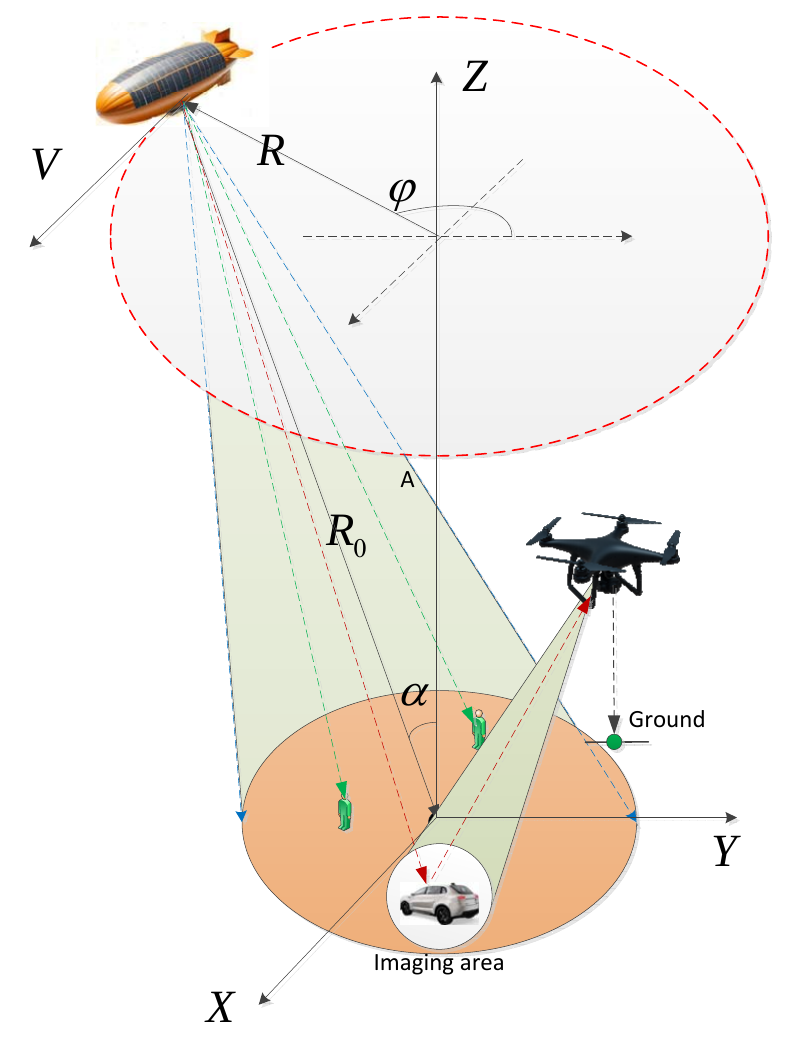}}
	\caption{The signal model for ISARAC in the context of a HAPS-UAV bistatic SAR system.}
	\label{fig1}	
\end{figure}
This section begins with an overall analysis of the HAPs–UAV bistatic ISARAC system architecture, highlighting its main components and their interaction mechanisms. Building upon this foundation, we then examine the system in detail from four perspectives, namely the UAV subsystem, the HAP subsystem, the SAR functionality, and the communication link.

Fig.\ref{fig1} illustrates the signal model of an ISARAC system based on a HAPs-UAV bistatic SAR architecture. In this system, a HAPs continuously transmits a downlink ISARAC waveform toward the ground. This waveform is carefully co-designed to simultaneously support both communication and radar sensing functionalities \cite{huang2025design}.
On the communication side, ground users directly receive the downlink signal from the HAPs and decode the embedded data, enabling wide-area coverage in non-terrestrial environments. Meanwhile, part of the transmitted signal is reflected by natural ground targets or environmental features, such as buildings, terrain, and vegetation, and these echoes are captured by a low-altitude UAV. Acting as a passive receiver, the UAV leverages the reflected signals to perform SAR imaging, thereby achieving high-resolution sensing and dynamic monitoring of the ground scene. Thanks to the spatial separation between the HAPS and UAV, this system forms a typical bistatic SAR geometry. By exploiting the UAV’s high mobility, rapid imaging can be performed over the HAPs-illuminated region, which is particularly valuable for emergency remote sensing applications such as post-earthquake damage assessment. Furthermore, since the UAV only receives echoes and does not transmit, its power consumption is significantly reduced, enabling long-duration and low-power imaging operations.

 \subsection{UAV analysis}
Assume the UAV operates over a synthetic aperture duration of $T$ seconds, during which azimuth sampling yields a total of $N_a$ discrete time slots, indexed by $n = 1, 2, \ldots, N_a$, with a uniform sampling interval $\delta$. Let $
\mathbf{q}_n=\left[ x_n,y_n,z_n \right] ^T\in \mathbb{R} ^3,\forall n
$ denote the UAV’s 3D position at time slot $n$, where the initial position is given by $
\mathbf{q}_0=\left[ x_0,y_0,z_0 \right] ^T
$. Note that $(\cdot)^T$ denotes the transpose operator of a matrix or vector.
The UAV's instantaneous velocity at time slot $n$ is represented by $
V_n
$.
 Consequently, the UAV’s mobility constraint can be formulated as
\begin{equation}
\label{eq1}
\parallel \mathbf{q}_{n}-\mathbf{q}_{n-1}\parallel =V_n\delta \le V\delta ,\quad \forall n,
\end{equation}
where $\|\cdot\|$ denotes the norm of a vector.
\eqref{eq1} ensures that the UAV's displacement between adjacent time slots does not exceed the distance permitted by its speed limit $V$.

Although the UAV operates as a passive receiver, a directional antenna can be employed to maximize the reception of reflected signals. Under this configuration, the ground location corresponding to the center of the UAV's antenna main lobe at time slot $n$ can be expressed as
\begin{equation}
    \label{}
\mathbf{q}_{n}^{c}=\left[ x_{n}^{c},y_{n}^{c},0 \right] ^T
\end{equation}
with
\begin{equation}
    \label{eq3}
x_{n}^{c}=x_n+z_n\tan \alpha \sin \gamma,\quad
y_{n}^{c}=y_n-z_n\tan \alpha \cos \gamma
\end{equation}
where  $
\gamma\in \left( 0,2\pi \right] 
$ represents the angle between the UAV’s flight direction and the horizontal $x$-axis. Besides, $\alpha\in \left( 0,\frac{\pi}{2} \right) $ represents the SAR observation angle, which is considered time-invariant during the entire synthetic aperture duration.

For a rotary-wing UAV, the propulsion power consumption at time $n$ when flying at velocity $V_n$ is expressed as \cite{zeng2019energy}
\begin{equation}
    \label{eq4}
    \begin{split}
      P_n=&P_0\left( 1+\frac{3V_{n}^{2}}{u_{\mathrm{tip}}^{2}} \right)    +\frac{1}{2}d_0\varrho sAV_{n}^{3} 
\\
\,\,   &\quad\quad\quad\quad\quad\quad+P_1\left( \sqrt{1+\frac{V_{n}^{4}}{4v_{0}^{4}}}-\frac{V_{n}^{2}}{2v_{0}^{2}} \right) ^{\frac{1}{2}}
    \end{split}
\end{equation}
where 
$P_0$ and $P_1$ represent the fixed blade profile power and induced power, respectively, required for the aircraft to maintain a hovering state.
$u_{\mathrm{tip}}$ denotes the rotor tip speed, while $v_0$ corresponds to the mean induced velocity during hover.
The parameters $d_0$ and $\varrho$ represent the fuselage drag coefficient and air density, respectively.
Meanwhile, $s$ and $A$ denote the rotor solidity and the rotor disc area, respectively.

\subsection{HAPs analysis}
Assuming the HAPs maintains a circular trajectory with a constant altitude $H$ and a uniform linear velocity $V_h$, its position at time slot $n$ is given by
\begin{equation}
    \label{}
    \begin{split}
     \boldsymbol{p}_n=&\left[ x_{p}^{n},y_{p}^{n},H \right] ^T
\\
     =&\left[ r\cos \left( \theta _0+\frac{n\delta V_h}{r} \right) ,r\sin \left( \theta _0+\frac{n\delta V_h}{r} \right) ,H \right] ^T
    \end{split}
\end{equation}
where $r$ denotes the radius and $\theta _0$ is the initial angle. In this respect, the grazing angle $\beta _n$ of the transmitted beam can be written as 
\begin{equation}
    \label{}
\cos \beta _n=\frac{\sqrt{\left\| x_{p}^{n}-x_{c}^{n} \right\| ^2+\left\| y_{p}^{n}-y_{c}^{n} \right\| ^2}}{\left\| \boldsymbol{p}_n-\mathbf{q}_{c}^{n} \right\|}
\end{equation}

Unlike UAV, HAPs are typically equipped with solar panels, enabling them to harvest solar energy to sustain their flight trajectory and support sensing and communication tasks. 
Therefore, in this paper, the propulsion-related energy consumption of HAPs, 
including platform flight maintenance, attitude control, and thermal management, 
is not explicitly modeled. 
This modeling choice is primarily made to maintain analytical tractability and to 
enable a clean mathematical formulation of the joint optimization problem. While 
propulsion- and platform-related energy expenditures are essential for practical 
HAPs operation, explicitly modeling them would introduce additional system-specific 
parameters and constraints that are beyond the scope of this study. 
By focusing on the ISARAC transmission energy, the proposed framework allows for a 
transparent and tractable analysis of the key design tradeoffs of interest. For 
more comprehensive HAPs energy consumption models that explicitly account for 
propulsion-related power as well as battery charging and discharging processes, 
interested readers are referred to \cite{huang2025joint,javed2023interdisciplinary}.


\subsection{SAR analysis}
To generate a high-quality bistatic SAR image using the HAPs-UAV system, it is crucial to ensure that the received echo signal at the UAV satisfies a minimum SNR requirement. It is important to emphasize that, unlike in communication systems or other radar modalities, metrics such as the signal-to-clutter ratio (SCR) or signal-to-interference-plus-noise ratio (SINR) are not considered here. This is because, in the context of SAR imaging, clutter and certain types of interference may actually contribute positively to the image formation process rather than degrading it. Certainly, ensuring a sufficient SNR at the receiver is only the initial step. To obtain the final SAR image, dedicated imaging algorithms must be employed to perform pulse compression along both the fast-time and slow-time dimensions \cite{cumming2005digital}. However, the detailed imaging procedure lies beyond the scope of this work.

According to the radar equation \cite{richards2005fundamentals}, it is easy to derive the expression of SNR at UAV receiver \cite{rodriguez2009bistatic}, as 
\begin{equation}
    \label{eq9}
\mathrm{SNR}\left( n \right) =\frac{P_t(n)G_tG_r\lambda ^3\sigma _0A_p}{\left( 4\pi \right) ^3\left\| \boldsymbol{p}_n-\mathbf{q}_{c}^{n} \right\| ^2\left\| \mathbf{q}_n-\mathbf{q}_{c}^{n} \right\| ^2KT_bF_nB_rL_s}
\end{equation}
where all symbols are given the meaning in the table \ref{tab1}.
\begin{table}[htp]
	\centering
	\caption{Symbols  used in \eqref{eq9}}
	\label{tab1}
	\begin{tabular}{c|p{3cm}cc}
		\hline
		{Symbol} &{Explanation} & Symbol & Explanation
		\\	\cline{1-4}	
		$P_t(n)$	& transmit power in slot $n$ &$G_t$ &	transmit gain for HAPs
        		\\	
		$G_r$	& receive gain for UAV &$\lambda$ &	wavelength
		\\
        $\sigma _0$ & radar cross section (RCS) &$L_s$ &loss factor for system 
        \\
        $K$ &Boltzmann’s constant &$T_b$ & receiver temperature\\
        $F_n$ &receiver noise figure &$B_r$& receiver's bandwidth
        \\
        $A_p$ &processing gain related to pulse compression and coherent time&&\\
		\hline
	\end{tabular}
	\label{tab1}
\end{table}

Furthermore,  this paper also uses spatial resolution serving as a critical indicator for evaluating the impact of different UAV trajectory designs on sensing accuracy. In this respect, the expressions for range resolution $
\delta _{n}^{r}
$ and azimuth resolution $
\delta _{n}^{a}
$ at time slot $n$ for HAPs-UAV bistatic SAR are given by \cite{moccia2011spatial,dower2018bistatic}
\begin{align}
\delta _{n}^{r}=&\frac{c}{B\left( \sin \alpha +\cos \beta _n \right)}
\\
\delta _{n}^{a}=&\frac{\lambda z_n}{T_wV_n\cos \alpha}
\end{align}
where the symbol $c$ denotes the light speed and $B$ is the bandwidth of transmitting signal from HAPs. $\lambda$  and $T_w$ denote the wavelength and the coherent processing interval (CPI) of the SAR signal, corresponding to one synthetic aperture duration, respectively.

The resolution of SAR is defined as the minimum distance at which two adjacent targets can be distinguished in a 2D SAR image. It typically includes range resolution $
\delta _{n}^{r}\leqslant d_{\min}^{r}
$ and azimuth resolution $
\delta _{n}^{a}\leqslant d_{\min}^{a}
$. Note that $d_{\min}^{r}$ and $d_{\min}^{a}$ denote the minimum resolvable distances between targets along the range and azimuth dimensions, respectively.
Conventional SAR systems enhance resolution by adjusting parameters such as signal bandwidth and antenna aperture, thereby improving target distinguishability \cite{hein2003processing}. However, it is important to note that this work does not focus on target recognition or imaging performance. Instead, our primary objective is to study the optimization of UAV trajectories. Hence, in future work, target distinguishability will be modeled by assigning varying resolution parameters, while abstracting from the detailed influence of signal bandwidth and radar system architecture on the achievable resolution. 
\subsection{Communication analysis}
This work focuses on the communication link between the HAPs and ground users. To facilitate analytical modeling, we assume that the ground communication user is located within the area imaged by the UAV. In particular, for mathematical tractability, it is further assumed that a representative communication user is positioned at the center of the UAV's directional beam.

To simplify the analysis, we assume that the channel between the HAPs and the ground user is dominated by a line-of-sight (LoS) component, and that the air-to-ground propagation adheres to the free-space path loss model. Under this assumption, the channel gain at time slot $n$ can be expressed as
\begin{equation}
    \label{eq:channel_gain}
\rho _n=\frac{\rho _0}{\left\| \boldsymbol{p}_n-\mathbf{q}_{c}^{n} \right\| ^2}
\end{equation}
where $\rho _0$ denotes the channel gain in one meter.

\begin{remark}
The free-space path loss model in \eqref{eq:channel_gain} captures the dominant 
large-scale attenuation in LoS-dominated HAPs--ground links and enables analytical 
tractability in the joint trajectory and resource optimization.

In practical environments, additional effects such as atmospheric absorption or 
rain attenuation (especially at high-frequency bands), residual scattering, and 
shadowing may also be present. These effects can be incorporated by augmenting the 
large-scale path loss with an additional attenuation factor and/or a composite 
fading coefficient. Specifically, the channel power gain in 
\eqref{eq:channel_gain} can be generalized as
\begin{equation}
\rho_n = \frac{\rho_0}{\left\| \mathbf{p}_n - \mathbf{q}_c^n \right\|^2}
\cdot 10^{-\frac{L_{\mathrm{atm}}(f,\theta_n)+L_{\mathrm{add}}}{10}}
\cdot |h_n|^2,
\end{equation}
where $L_{\mathrm{atm}}(f,\theta_n)$ accounts for atmospheric or rain attenuation 
depending on the carrier frequency $f$ and elevation angle $\theta_n$, 
$L_{\mathrm{add}}$ represents other deterministic losses, and $h_n$ models 
small-scale fading, such as Rician fading for LoS-dominant air-to-ground channels.

Importantly, the proposed optimization framework remains applicable under this 
generalized channel model, since the additional effects mainly modify the 
effective channel gain and thus the required transmit power to satisfy a given 
QoS constraint, without altering the problem structure.
\end{remark}

Furthermore, the received SNR in ground user is 
\begin{equation}
    \label{}
\psi _n=\frac{P_t(n)G_tG_{r}^{c}\rho _n}{\tilde{\sigma}_{0}^{2}}
\end{equation}
where ${\tilde{\sigma}_{0}^{2}}$ denotes the noise power and $G_{r}^{c}$ is the receiver gain for communication user. Consequently, the final expression for the achievable communication rate can be derived as
\begin{equation}
    \label{}
R\left( n \right)=B\log \left( 1+\psi _n \right) 
\end{equation}
where $B$ is the bandwidth for communication link. 

\section{Proposed framework for HAPs-UAV bistatic ISARAC}
\label{sec3}
This section first formulates a joint SAR sensing–communication optimization model tailored for the HAPs–UAV bistatic scenario. Then, we develop a solution algorithm based on AO to address the UAV 3D trajectory planning and coupled sensing/communication decision-making problem in practical deployments.
\subsection{Problem Formulation}
This paper considers a bistatic ISARAC scenario, where the HAPs transmits ISARAC signals to serve ground users, while a UAV, exploiting its mobility, performs SAR sensing by receiving the reflected echoes of the HAPs’ signals from ground targets. Since communication remains the primary objective, the goal is to maximize the achievable communication sum rate under UAV sensing and energy constraints, namely\footnote{The proposed formulation assumes nominal platform positions and channel states. 
In practice, UAV positioning errors and HAPs trajectory drift may cause bistatic range-history mismatch and motion-induced phase errors in SAR imaging, which can be partly mitigated by motion compensation, calibration, and autofocus processing. 
Time-varying channel uncertainties mainly affect the effective communication channel gain and achievable rate, and robust trajectory--power optimization under such uncertainties is left for future work.}
\begin{align}
    \label{}
\begin{split}
    &
\left( \mathcal{P} 1 \right) :
\underset{\mathbf{q}_n,V_n,P_t(n)}{\max}\,\,
\sum_{n=1}^N{R\left( n \right)}
\\
&s.t. C1:\parallel \mathbf{q}_n-\mathbf{q}_{n-1}\parallel \leqslant V_n\delta ,\quad \forall n,
\\
& C2:0\leqslant V_n\leqslant V,\forall n,
\\
&C3:\sum_{n=1}^N{P_n\delta}<E,
\\
&C4:\mathrm{SNR}\left( n \right) \geqslant \mathrm{SNR}_{\min},
\\
&C5:\frac{c}{B\left( \sin \alpha +\cos \beta _n \right)}\leqslant d_{\min}^{r},
\\
&C6:\frac{\lambda z_n}{T_wV_n\cos \alpha}\leqslant d_{\min}^{a},\\
&C7:
z_{\min}\leqslant z_n\leqslant z_{\max},\forall n,\\
&C8:
0\leqslant P_t(n)\leqslant P_{\max},\forall n,
\\
&C9:
\sum_{n=1}^N{P_t(n)\delta}\leqslant E^{\mathrm{s}},
\end{split}
\end{align}
where constraints $C1$ and $C2$ respectively impose limits on the UAV’s maximum displacement and velocity, ensuring that its trajectory remains physically feasible. Constraint $C3$ represents the UAV’s energy limitation, modeled based on its initial onboard battery energy $E$. It is important to note that, since the SAR onboard the UAV is not assumed to perform real-time imaging in this paper, only the propulsion-related energy consumption is considered, while other sources of energy consumption such as communication and onboard processing are neglected. Constraint $C4$ sets a minimum SNR, namely  $\mathrm{SNR}_{\min}$, requirement for SAR imaging to ensure sufficient image quality. Constraints $C5$ and $C6$ define the minimum resolution requirements in the range and azimuth directions, respectively, to guarantee adequate spatial resolution for ground target observation. Constraint $C7$ specifies the feasible altitude range for the UAV, whereas Constraint $C8$ ensures that the HAPs’ transmit power is nonnegative and does not exceed its maximum limit. Although the present work does not explicitly account for the energy consumption of HAPs, it should be noted that, due to the platform’s inherent characteristics, energy allocation must be carefully pre-planned across different subsystems. In practice, ensuring the safe and sustained flight of HAPs has the highest priority, while sensing and communication tasks can only be supported within the remaining energy budget. Based on this rationale, we impose constraint $C9$ to explicitly capture the upper limit on the energy available for ISARAC operations.

It is important to emphasize that problem $\left( \mathcal{P} 1 \right)$ is inherently non-convex. This is due to the fact that its objective function, as well as constraints $C3$, $C4$, and $C5$, are all non-convex with respect to the parameter $\{\mathbf{q}_n, V_n, P_t(n)\}$. To address this challenge, we next introduce a tailored optimization algorithm grounded in the principle of AO, enabling efficient handling of the coupled variables.

\subsection{AO-based solution}
Since the objective function of problem $\left( \mathcal{P}1 \right)$ is neither convex nor concave with respect to the joint optimization variables $
\left\{ \mathbf{q}_{n}^{},P_t(n) \right\} 
$, directly solving it is intractable. To address this challenge, we adopt an AO strategy to decouple the original problem into two more manageable subproblems, which are solved iteratively.

Specifically, we first fix the UAV trajectory $
\left\{ \mathbf{q}_{n}^{} \right\} ,\forall n
$ and optimize the transmit power $
\left\{ P_t(n) \right\} ,\forall n
$ accordingly. Given that the transmit power subproblem often has a convex structure under fixed trajectory, it can be efficiently solved using standard convex optimization tools.

Next, with the optimized transmit power $P_{t}(n)$ obtained from the previous step, we proceed to update the UAV trajectory $\mathbf{q}_{n}$. Although the trajectory subproblem remains non-convex, its objective and constraints have exploitable structures. 
We therefore use first-order Taylor approximation and successive convex approximation (SCA) to construct a convex surrogate problem, which is solved iteratively to update the UAV trajectory.

By alternating between these two subproblems, power allocation and trajectory optimization, we obtain a suboptimal but efficient solution to the original non-convex problem. This AO-based framework ensures that each iteration monotonically improves the objective value, and convergence to a locally optimal solution is typically guaranteed under mild conditions.

Specifically, by fixing the UAV trajectory $
\left\{ \mathbf{q}_{n}^{\left( k \right)} \right\} ,\forall n
$ in the $k$th iteration, the original non-convex problem $\left( \mathcal{P}1 \right)$ can be transformed into
\begin{equation}
    \label{}
\begin{aligned}
	\left( \mathcal{P} 2 \right) :&\sum_{n=1}^N{\underset{P_t(n)}{\max}\,\,B\log \left( 1+P_t(n)\frac{G_tG_{r}^{c}\rho _0}{\left\| \boldsymbol{p}_n-\mathbf{q}_{c}^{n,\left( k \right)} \right\| ^2\sigma _{0}^{2}} \right)}\\
	&s.t.C4,C8,C9\\
\end{aligned}
\end{equation}
where the fixed ground user location $\mathbf{q}_{c}^{n,\left( k \right)}$ of the UAV antenna phase center, is computed according to \eqref{eq3} by inserting $
\mathbf{q}_{n}^{\left( k \right)}
$. Moreover, it is straightforward to verify that problem $\left( \mathcal{P}2 \right)$ is convex.

Although subproblem $\left( \mathcal{P}2 \right)$ can be solved using optimization tools such as CVX, this paper seeks to present an alternative approach that offers improved numerical stability. 
In particular, Proposition 1 enables us to derive the closed-form solution $P_{t}^{\star}(n)$ for $P_{t}(n)$.

\begin{proposition}[Capped water-filling]
\label{prop:cwf}
Fix $\{\boldsymbol q_n^{(k)}\}$ and $
\left\{ \boldsymbol{q}_{c}^{n,(k)} \right\} 
$.
Let
\begin{align}
    \label{}
a_{n}^{\left( k \right)}=&\frac{G_tG_{r}^{c}\rho _0}{\left\| \boldsymbol{p}_n-\mathbf{q}_{c}^{n,\left( k \right)} \right\| ^2\sigma _{0}^{2}}>0,\\
\begin{split}
\tilde{L}_{n}^{\left( k \right)}
=&\parallel \boldsymbol{p}_n-\boldsymbol{q}_{c}^{n,(k)}\parallel ^2\parallel \boldsymbol{q}_{n}^{(k)}-\boldsymbol{q}_{c}^{n,(k)}\parallel ^2
\\
\,\,       &\times \frac{\mathrm{SNR}_{\min}(4\pi )^3KT_bF_nB_rL_s}{G_tG_r\lambda ^3\sigma _0A_p}
\end{split}
\end{align}
and denote the clamping operator $[x]_{L}^{U} \triangleq \min\{\max\{x,L\},U\}$. 
Next, consider $(\mathcal P2)$ like 
\begin{align}
\left( \widetilde{\mathcal{P} 2} \right): &
\max_{\{P_t(n)\}} ~\sum_{n=1}^N{B\log \!\bigl( 1+a_{n}^{(k)}P_t(n) \bigr) \,\,}
\\
\mathrm{s}.\mathrm{t}. &\tilde{L}_{n}^{(k)}\le P_t(n)\le P_{\max},
\\
\,\,    &\sum_{n=1}^N{P_t(n)\delta}\le E^{\mathrm{s}},
\end{align}
If $\sum_{n=1}^N \tilde L_n^{(k)}\,\delta \le E^{\mathrm{s}}$ and $\tilde L_n^{(k)}\le P_{\max}$  for all $n$, the problem is feasible and its optimal solution admits
\begin{equation}
\label{eq:ptstar}
P_t^{\star}(n;\mu)
= \Big[\ \mu-\tfrac{1}{a_n^{(k)}}\ \Big]_{\tilde L_n^{(k)}}^{P_{\max}},\qquad n=1,\dots,N,
\end{equation}
where the water level $\mu>0$ is determined as follows. 

Define $E_{\mathrm{cap}}\triangleq E^{\mathrm{s}}/\delta$ and
\begin{equation}
S(\mu)\triangleq \sum_{n=1}^N \Big[\ \mu-\tfrac{1}{a_n^{(k)}}\ \Big]_{\tilde L_n^{(k)}}^{P_{\max}} .    
\end{equation}

Then $S(\mu)$ is continuous, piecewise linear, and nondecreasing in $\mu$ (strictly increasing whenever at least one slot is unsaturated). Hence, we have\footnote{By the monotonicity of $S(\mu)$, the water level $\mu$ can be found via a bisection search in $O(N\log(1/\varepsilon))$ time \cite{boyd2004convex,palomar2006tutorial}.
}
\begin{enumerate}
\item If $
\sum_{n=1}^N{\tilde{L}_{n}^{(k)}}\le E_{\mathrm{cap}}\le \sum_{n=1}^N{P_{\max}}
$, there exists a unique $\mu>0$ such that $S(\mu)=E_{\mathrm{cap}}$, i.e.,
$\sum_{n=1}^{N} P_t^{\star}(n;\mu)\,\delta = E$ (the energy constraint is active).
\item If $
E_{\mathrm{cap}}\geqslant \sum_{n=1}^N{P_{\max}}
$, then $P_t^{\star}(n)=P_{\max}$ for all $n$ and 
$
\sum_{n=1}^N{P_{t}^{\star}(n )\delta}=\sum_{n=1}^N{P_{\max}\delta}\le E^{\mathrm{s}}
$ (slack).
\item If $
E_{\mathrm{cap}}\le \sum_{n=1}^N{\tilde{L}_{n}^{(k)}}
$, $(\mathcal P2)$ is infeasible.
\end{enumerate}
\end{proposition}

\begin{proof}
The objective $
\sum_{n=1}^N{B\log\mathrm{(}1}+a_{n}^{(k)}P_t(n))
$ is strictly concave in $\{P_t(n)\}$ and the constraints are affine. Thus $\left( \widetilde{\mathcal{P} 2} \right)$ is a convex program. 
When $\sum_n \tilde L_n^{(k)}\,\delta < E$ and $\tilde L_n^{(k)}<P_{\max}$ for all $n$,
Slater's condition holds, so Karush–Kuhn–Tucker (KKT) conditions are necessary and sufficient \cite{rockafellar2015convex}.

Form the Lagrangian expression \cite{boyd2004convex}
\begin{equation}
    \label{}
    \begin{split}
\mathcal{L} =&\sum_{n=1}^N{B\log\mathrm{(}1}+a_{n}^{(k)}P_t(n))-\tilde{\lambda}\bigl( \sum_{n=1}^N{P_t(n)\delta}-E^{\mathrm{s}} \bigr) 
\\
\,\,    &+\sum_{n=1}^N{\nu _n\bigl( P_t(n)-\tilde{L}_{n}^{(k)} \bigr)}+\sum_{n=1}^N{\xi _n\bigl( P_{\max}-P_t(n) \bigr) ,}
    \end{split}
\end{equation}
with multipliers $\tilde{\lambda}\ge0$, $\nu_n\ge0$, and $\xi_n\ge0$.
The stationarity conditions read
\begin{equation}
    \label{}
    \frac{\partial \mathcal L}{\partial P_t(n)}
= \frac{B a_n^{(k)}}{1+a_n^{(k)}P_t^{\star}(n)} - \lambda\delta + \nu_n - \xi_n = 0,\quad n=1,\dots,N.
\end{equation}

Complementary slackness gives $\nu_n\big(P_t^{\star}(n)-\tilde L_n^{(k)}\big)=0$ and $\xi_n\big(P_{\max}-P_t^{\star}(n)\big)=0$.
Hence, if $\tilde L_n^{(k)}<P_t^{\star}(n)<P_{\max}$, then $\nu_n=\xi_n=0$ and
\begin{equation}
    \label{}
\frac{Ba_{n}^{(k)}}{1+a_{n}^{(k)}P_{t}^{\star}(n)}=\tilde{\lambda}\delta \,\,\Rightarrow \mathrm{
}P_{t}^{\star}(n)=\frac{B}{\tilde{\lambda}\delta}-\frac{1}{a_{n}^{(k)}}.
\end{equation}
If $P_t^\star(n)=\tilde L_n^{(k)}$, the expression is clipped at the lower bound. Meanwhile, if $P_t^\star(n)=P_{\max}$, it is clipped at the upper bound.
Letting $\mu\triangleq B/(\lambda\delta)$ yields \eqref{eq:ptstar}.

To determine $\mu$, define $E_{\mathrm{cap}}\triangleq E^{\mathrm{s}}/\delta$ and
$S(\mu)\triangleq \sum_{n=1}^N P_t^{\star}(n;\mu)$.
Since $P_t^{\star}(n;\mu)=[\,\mu-1/a_n^{(k)}\,]_{\tilde L_n^{(k)}}^{P_{\max}}$,
$S(\mu)$ is continuous and nondecreasing in $\mu$, with
\begin{equation}
\label{}
\lim_{\mu\to-\infty} S(\mu)=\sum_{n=1}^N \tilde L_n^{(k)},
\quad
\lim_{\mu\to+\infty} S(\mu)=\sum_{n=1}^N P_{\max}.
\end{equation}

Hence, if $
\sum_{n=1}^N{\tilde{L}_{n}^{(k)}}\le E_{\mathrm{cap}}\le \sum_{n=1}^N{P_{\max}}
$,
there exists a unique $\mu>0$ such that $S(\mu)=E_{\mathrm{cap}}$, i.e.,
$\sum_{n=1}^N P_t^{\star}(n;\mu)\,\delta=E^{\mathrm{s}}$ (the energy constraint is active).
If $E_{\mathrm{cap}} \ge \sum_n P_{\max}$, then $P_t^{\star}(n)=P_{\max}$ for all $n$ and
$\sum_{n=1}^N P_t^{\star}(n)\,\delta=\sum_{n=1}^N P_{\max}\,\delta \le E^{\mathrm{s}}$ (slack).
By complementary slackness, the optimal dual multiplier is $\tilde{\lambda}^\star=0$ (thus the
unconstrained expression $B/(
\tilde{\lambda}
^\star\delta)$ diverges, which is consistent with saturation at $P_{\max}$).
If $E_{\mathrm{cap}}<\sum_{n=1}^N \tilde L_n^{(k)}$, $(\mathcal P2)$ is infeasible.

This completes the proof.
\end{proof}

Next, we optimize the UAV trajectory based on the transmit power $
P_{t}^{\left( k \right)}(n)
$ obtained from the previous step. Accordingly, from the original problem $\left( \mathcal{P}1 \right)$, we derive another equivalent subproblem that focuses on trajectory optimization based on the fixed $
P_{t}^{\left( k \right)}(n)
$, namely
\begin{equation}
    \label{}
\begin{aligned}
	\left( \mathcal{P} 3 \right) :&\sum_{n=1}^N{\underset{V_n,\mathbf{q}_{n}^{}}{\max}\,\,B\log \left( 1+P_{t}^{\left( k \right)}(n)\frac{G_tG_{r}^{c}\rho _0}{\left\| \boldsymbol{p}_n-\mathbf{q}_{c}^{n} \right\| ^2\sigma _{0}^{2}} \right)}\\
	&s.t.C1-C7\\
\end{aligned}
\end{equation}

Notice  that the objective function $
B\log \left( 1+P_{t}^{\left( k \right)}(n)\frac{G_tG_{r}^{c}\rho _0}{\left\| \boldsymbol{p}_n-\mathbf{q}_{c}^{n} \right\| ^2\sigma _{0}^{2}} \right) 
$ is not concave with respect to  $\mathbf{q}_{n}^{}$, but is convex with respect to $
\left\| \boldsymbol{p}_n-\mathbf{q}_{c}^{n} \right\| ^2
$. Since  a convex function is globally bounded below by its first-order Taylor approximation at any point, we construct a lower bound for the objective around the UAV trajectory at $k$th iteration as
\begin{equation}
    \label{}
    \begin{split}
&B\log \left( 1+P_{t}^{\left( k \right)}(n)\frac{G_tG_{r}^{c}\rho _0}{\left\| \boldsymbol{p}_n-\mathbf{q}_{c}^{n} \right\| ^2\sigma _{0}^{2}} \right) \geqslant \\&\mathcal{M} ^{\left( k \right)}\left( n \right) -\mathcal{N} ^{\left( k \right)}\left( n \right) \left[ \left\| \boldsymbol{p}_n-\mathbf{q}_{c}^{n} \right\| ^2-\left\| \boldsymbol{p}_n-\mathbf{q}_{c}^{n,\left( k \right)} \right\| ^2 \right] \\
&\triangleq \tilde{R}\left( \mathbf{q}_{n}^{} \right) 
    \end{split}
\end{equation}
where both $
\mathcal{M} ^{\left( k \right)}\left( n \right) 
$ and $
\mathcal{N} ^{\left( k \right)}\left( n \right) 
$ are constant, which can be derived as

\begin{align}
      \label{}
&\mathcal{M} ^{\left( k \right)}\left( n \right) =B\log \left( 1+P_{t}^{\left( k \right)}(n)\frac{G_tG_{r}^{c}\rho _0}{\left\| \boldsymbol{p}_n-\mathbf{q}_{c}^{n,\left( k \right)} \right\| ^2\sigma _{0}^{2}} \right) 
\\
&\mathcal{N} ^{\left( k \right)}\left( n \right) =\frac{BP_{t}^{\left( k \right)}(n)G_tG_{r}^{c}\rho _0\log e}{\left\| \boldsymbol{p}_n-\mathbf{q}_{c}^{n,\left( k \right)} \right\| ^4\sigma _{0}^{2}+\left\| \boldsymbol{p}_n-\mathbf{q}_{c}^{n,\left( k \right)} \right\| ^2P_{t}^{\left( k \right)}(n)G_tG_{r}^{c}\rho _0}  
\end{align}

Hence, the lower bound $
\tilde{R}\left( \mathbf{q}_{n}^{} \right) 
$ is concave with respect to $\mathbf{q}_{n}^{}$ and the subproblem $\left( \mathcal{P} 3 \right)$ can be approximated as
\begin{equation}
    \label{}
\begin{aligned}
	\left( \mathcal{P} 4 \right) :&\sum_{n=1}^N{\underset{V_n,\mathbf{q}_{n}^{}}{\max}\,\,\tilde{R}\left( \mathbf{q}_{n}^{} \right)}\\
	&s.t.C1-C7\\
\end{aligned}
\end{equation}

Noted that subproblem $\left( \mathcal{P} 4 \right)$ remains non‐convex, since constraints $C3$, $C4$, $C5$ are non‐convex. Next, we will systematically reformulate each of these constraints into convex form.

To begin, we recast $C3$ as a convex constraint. Accordingly, define a new 
auxiliary variables $
\boldsymbol{h}=\left[ h_1,h_2,\cdots ,h_N \right] ^T\in \mathbb{R} ^N
$ with $
 h_{n}^{}>0, \forall n 
$, namely
\begin{equation}
    \label{eq20}
h_{n}^{2}=\sqrt{1+\frac{V_{n}^{4}}{4v_{0}^{4}}}-\frac{V_{n}^{2}}{2v_{0}^{2}}
\end{equation}
Moreover, by a direct algebraic manipulation, one can readily establish the following identity
\begin{equation}
    \label{eq21}
\frac{1}{h_{n}^{2}}=h_{n}^{2}+\frac{V_{n}^{2}}{v_{0}^{2}}
\end{equation}

Inserting \eqref{eq20} into \eqref{eq4} yields 
\begin{equation}
    \label{}
\widetilde {P_n}
\triangleq 
P_0\left( 1+\frac{3V_{n}^{2}}{u_{\mathrm{tip}}^{2}} \right) +\frac{1}{2}d_0\varrho sAV_{n}^{3}+P_1h_{n}^{}
\end{equation}

It is straightforward to verify that $\widetilde{P}_n$ is jointly convex in the variables $(V_{n},\,h_{n})$. Consequently, we can further reformulate subproblem $\left( \mathcal{P} 4 \right) $ into $\left( \mathcal{P} 5 \right) $, namely
 \begin{equation}
    \label{}
\begin{split}
\left( \mathcal{P} 5 \right) :&
\underset{\mathbf{q}_n,V_n}{\max}\,\, \widetilde R\left( \mathbf{q}_{n}^{} \right)
\\
s.t.& C1-C2,\\
&
C3a:\sum_{n=1}^N{\widetilde{P_n}\delta}<E,
\\
&C3b:\frac{1}{h_{n}^{2}}\leqslant h_{n}^{2}+\frac{V_{n}^{2}}{v_{0}^{2}},
\\
& C4-C7.
\end{split}
\end{equation}

Besides, constraint $C3b$ is not convex.  However, it is easy to see that $C3b$ can be written as the difference of convex (DC) of two convex functions, namely $
g_{3}(h_{n}) = \frac{1}{h_{n}^{2}}$
$
 g_{4}(h_{n},V_{n}) = h_{n}^{2} + \frac{V_{n}^{2}}{v_{0}^{2}}.
$
Then, the first‐order Taylor expansion of $ g_{4}(p_{n},V_{n})$ at the current point $\bigl(h_{n}^{(k)},V_{n}^{(k)}\bigr)$ is given by
\begin{equation}
    \label{}
    \begin{split}
       {g_4}\left( h_{n}^{},V_{n}^{} \right) 
\geqslant 
&h_{n}^{\left( k \right) ^2}+\frac{V_{n}^{\left( k \right) ^2}}{v_{0}^{2}}+2h_{n}^{\left( k \right)}\left( h_{n}^{}-h_{n}^{\left( k \right)} \right) 
\\
&\quad\quad\quad\quad\quad\quad+\frac{2V_{n}^{\left( k \right)}}{v_{0}^{2}}\left( V_{n}^{}-V_{n}^{\left( k \right)} \right)\\
\triangleq &\widetilde{g_4}.  
    \end{split}
\end{equation}

By using the lower-bound approximation $\widetilde{g_4}\left( h_{n}^{},V_{n}^{} \right)$ in place of ${g_4}\left( h_{n}^{},V_{n}^{} \right)$, constraint $C3b$ is rendered convex. Consequently, subproblem $\left( \mathcal{P} 5 \right) $ can be equivalently formulated as 
 \begin{equation}
    \label{}
\begin{split}
\left( \mathcal{P} 6 \right) :&
\underset{\mathbf{q}_n,V_n}{\max}\,\, \widetilde R\left( \mathbf{q}_{n}^{} \right)
\\
s.t.& C1-C2,\\
&
C3a:\sum_{n=1}^N{\widetilde{P_n}\delta}<E,
\\
&
\widetilde{C3b}:g_3\left( h_{n}^{} \right) -\widetilde{g_4}\left( h_{n}^{},V_{n}^{} \right) \leqslant 0,
\\
& C4-C7.
\end{split}
\end{equation}

However, due to the non-convexity of constraint $C4$, the suboptimization problem $\left( \mathcal{P} 6 \right) $ remains non-convex. Therefore, we proceed to transform the constraint  $C4$ into a convex constraint. 
In this respect, $C4$ can be further expressed as 
\begin{equation}
    \label{eqq22}
\left\| \mathbf{q}_n-\mathbf{q}_{c}^{n} \right\| \left\| \boldsymbol{p}_n-\mathbf{q}_{c}^{n} \right\| \leqslant \sqrt{P_{t}^{\left( k \right)}(n)\zeta}
\end{equation}
with the expression of $
\zeta
$ being $
\zeta =\frac{G_tG_r\lambda ^3\sigma _0A_p}{\left( 4\pi \right) ^3\mathrm{SNR}_{\min}KT_bF_nB_rL_s}
$. 

Next, denote 
\begin{align}
    \label{}
g_5\left( \mathbf{q}_n \right) =&\left\| \boldsymbol{p}_n-\mathbf{q}_{c}^{n} \right\| 
\\
g_6\left( \mathbf{q}_n \right) =&\left\| \mathbf{q}_n-\mathbf{q}_{c}^{n} \right\| =z_n\sqrt{1+\tan ^2\alpha}
\label{}
\end{align}

Note that $g_5\left( \mathbf{q}_n \right) $ and $g_6\left( \mathbf{q}_n \right)$ are all convex. Hence, the left-hand side of the inequality  in \eqref{eqq22} can be rewritten as 
\begin{equation}
    \label{eqq25}
g_5\left( \mathbf{q}_n \right) g_6\left( \mathbf{q}_n \right) =\left( g_5\left( \mathbf{q}_n \right) +g_6\left( \mathbf{q}_n \right) \right) ^2-\left( g_{5}^{2}\left( \mathbf{q}_n \right) +g_{6}^{2}\left( \mathbf{q}_n \right) \right) 
\end{equation}
In \eqref{eqq25}, the product $g_5\left( \mathbf{q}_n \right) g_6\left( \mathbf{q}_n \right)$ is represented as a DC functions. To facilitate convex approximation, we apply the first-order Taylor expansion to obtain convex lower bounds for $g_{5}^{2}\left( \mathbf{q}_n \right)$ and $g_{6}^{2}\left( \mathbf{q}_n \right)$ around the reference point $\mathbf{q}_{n}^{(k)}$, namely
\begin{align}
    \label{}
g_{5}^{2}\left( \mathbf{q}_n \right) &\geqslant g_{5}^{2}\left( \mathbf{q}_{n}^{(k)} \right) +\bigtriangledown g_{5}^{2}\left( \mathbf{q}_{n}^{(k)} \right) \left( \mathbf{q}_n-\mathbf{q}_{n}^{(k)} \right) \triangleq \widetilde{g_{5}^{2}}\left( \mathbf{q}_n \right) 
\\
g_{6}^{2}\left( \mathbf{q}_n \right) &\geqslant g_{6}^{2}\left( \mathbf{q}_{n}^{(k)} \right) +\bigtriangledown g_{6}^{2}\left( \mathbf{q}_{n}^{(k)} \right) \left( \mathbf{q}_n-\mathbf{q}_{n}^{(k)} \right) \triangleq \widetilde{g_{6}^{2}}\left( \mathbf{q}_n \right) 
\end{align}
where $\nabla f(x)$ denotes the first-order derivative of function $f(x)$ with respect to  variable $x$.
Moreover, the expression of $\bigtriangledown g_{5}^{2}\left( \mathbf{q}_{n}^{} \right) $ is 
\begin{align}
    \label{eqq28}
\nabla g_{5}^{2}(\mathbf{q}_n)=&\left[ \begin{array}{c}
	x_{n}^{c}-x_{p}^{n}\\
	y_{n}^{c}-y_{p}^{n}\\
	\tan \alpha \left[ (x_{n}^{c}-x_{p}^{n})\sin \gamma _n-(y_{n}^{c}-y_{p}^{n})\cos \gamma _n \right]\\
\end{array} \right] 
\\
\nabla g_{6}^{2}(\mathbf{q}_n)=&2\left( 1+\tan ^2\alpha \right) z_n
\end{align}
Interesting reader may refer to Appendix \ref{refA} for more details about the derivation of \eqref{eqq28}.

Now, the constraint $C4$ can be reformulated as
\begin{equation}
    \label{}
\widetilde{C4}:\left( g_5\left( \mathbf{q}_n \right) +g_6\left( \mathbf{q}_n \right) \right) ^2-\widetilde{g_{5}^{2}}\left( \mathbf{q}_n \right) -\widetilde{g_{6}^{2}}\left( \mathbf{q}_n \right) \leqslant \sqrt{P_{t}^{\left( k \right)}(n)\zeta}
\end{equation}


After the above transformations, the sub-optimization problem $\left( \mathcal{P}6 \right)$ can be equivalently expressed as 
\begin{equation}
    \label{}
\begin{aligned}
	\left( \mathcal{P} 7 \right) :&\underset{\mathbf{q}_n,V_n}{\max}\,\,\widetilde{R}\left( \mathbf{q}_{n}^{} \right)\\
	s.t.&C1-C2,C3a,\widetilde{C3b},\widetilde{C4},C5-C7.\\
\end{aligned}
\end{equation}

Upon examining the sub-optimization problem $\left( \mathcal{P}7 \right)$, it becomes evident that the non-convexity of constraint $C5$ continues to impede the convexity of the overall formulation. To overcome this challenge, we now focus on transforming constraint $C5$ into a convex form. In this respect, $C5$ can be rewritten as 
\begin{equation}
    \label{}
\cos \beta _n\geqslant \frac{c}{Bd_{\min}^{r}}-\sin \alpha 
\end{equation}
Notice that $
\cos ^2\beta _n=1-\sin ^2\beta _n
$, it is easily to obtain 
\begin{equation}
    \label{}
\widetilde{C5}:\frac{H^2}{\left\| \boldsymbol{p}_n-\mathbf{q}_{c}^{n} \right\| ^2}\leqslant 1-\left( \frac{c}{Bd_{\min}^{r}}-\sin \alpha \right) ^2
\end{equation}

Notice that the constraint $\widetilde{C5}$ is also non-convex. Further, to convexify it, we introduce two slack variables $
\boldsymbol{t}=\left[ t_1,t_2,\cdots ,t_N \right] ^T\in \mathbb{R} ^N$, $
\boldsymbol{u}=\left[ u_1,u_2,\cdots ,u_N \right] ^T\in \mathbb{R} ^N
$. Next, the constraint $\widetilde{C5}$ can equivalently be written as
\begin{align}
    \label{}
    \begin{split}
\exists u_n>0,\widetilde{C5a}:\,u_n\;\leqslant \;1-\Bigl( \tfrac{c}{B\,d_{\min}^{r}}-\sin \alpha \Bigr) ^2,\forall n\\
\exists t_n>0,\widetilde{C5b}:t_n\;\ge \;\left\| \boldsymbol{p}_n-\mathbf{q}_{c}^{n} \right\| ^2,\forall n
\\
\widetilde{C5c}:
\left[ \begin{matrix}
	u_n&		H\\
	H&		t_n\\
\end{matrix} \right] \succeq 0
    \end{split}
\end{align}
 At this point, the originally non-convex optimization problem $\left( \mathcal{P}3 \right)$ has been equivalently reformulated as a convex problem $\left( \mathcal{P}8 \right)$, which can be expressed as 
\begin{equation}
    \label{}
\begin{aligned}
	\left( \mathcal{P} 8 \right) :&\underset{\mathbf{q}_n,V_n}{\max}\,\,\widetilde{R}\left( \mathbf{q}_{n}^{} \right)\\
	s.t.&C1-C2,C3a,\widetilde{C3b},\widetilde{C4},\widetilde{C5a},\widetilde{C5b},\widetilde{C5c},C6-C7.\\
\end{aligned}
\end{equation}
The sub-optimization problem $\left( \mathcal{P}8 \right)$ can be efficiently solved using established optimization tools such as CVX \cite{grant2014cvx}. Moreover, the detailed procedure of the proposed optimization algorithm is outlined in Algorithm \ref{alg1}.

\begin{algorithm}
\caption{The proposed trajectory planning algoithm for HAPs-UAV bistatic ISARAC for problem $\left( \mathcal{P} 8 \right)$}
\label{alg1}
\begin{algorithmic}[1]
\STATE \textbf{Initialization:} initial the UAV path and power $
\left\{ \mathbf{q}_{n}^{\left( 0 \right)},V_{n}^{\left( 0 \right)},P_{t}^{\left( 0 \right)}(n),n=1,2,\cdots ,N \right\} 
$,  \textbf{Set} tolerance $
\delta =10^{-10}
$ and $
{R}^{\left( 0 \right)}=-\infty 
$;
\STATE obtain the  $
\left\{ h_{n}^{\left( 0 \right)},n=1,2,\cdots ,N \right\} 
$ via \eqref{eq21};
\REPEAT 
\STATE Set $k=k+1$
\STATE Get the power $
\left\{ P_{t}^{\left( k \right)}(n),n=1,2,\cdots ,N \right\} 
$ by solving  \( \left( \widetilde{\mathcal{P} 2} \right) \) around the points $
\left\{ \mathbf{q}_{n}^{\left( k-1 \right)},n=1,2,\cdots ,N \right\} 
$;
\STATE Obtain the UAV trajectory parameters $
\left\{ \mathbf{q}_{n}^{\left( k \right)},V_{n}^{\left( k \right)},h_{n}^{\left( k \right)},n=1,2,\cdots ,N \right\} 
$ by by solving  \( \left( \mathcal{P} 8 \right) \) in the fixed power $
\left\{ P_{t}^{\left( k \right)}(n),n=1,2,\cdots ,N \right\} 
$;

\STATE Calculate $
R^{\left( k \right)}
$ by inserting $
\left\{ \mathbf{q}_{n}^{\left( k \right)},V_{n}^{\left( k \right)},P_{t}^{\left( k \right)}(n),n=1,2,\cdots ,N \right\} 
$ into the objective of \( \left( \mathcal{P} 1 \right) \);

\STATE Computate $
\xi =\left| \frac{{R}^{\left( k \right)}-{R}^{\left( k-1 \right)}}{{R}^{\left( k \right)}} \right|
$
\UNTIL{$
\xi \leqslant \delta 
$}
\RETURN $
{R}^{\left( k \right)}
$
\end{algorithmic}
\end{algorithm}

\subsection{Computational Complexity Analysis}
\label{subsec:complexity}

We analyze the computational complexity of the proposed AO-based algorithm. 
Let $I_{\rm AO}$ denote the number of AO iterations, $\epsilon_{\rm wf}$ denote the bisection accuracy for determining the water level in the capped water-filling solution, and $I_{\rm IP}$ denote the number of interior-point iterations for solving the convex trajectory subproblem.

In each AO iteration, the proposed algorithm mainly consists of two steps, namely transmit-power allocation and UAV trajectory update. 
For the transmit-power allocation subproblem, the capped water-filling solution in Proposition~\ref{prop:cwf} is adopted. 
For the transmit-power allocation subproblem, the capped water-filling solution in Proposition~\ref{prop:cwf} is adopted. 
Given a water level $\mu$, the transmit power $P_t^\star(n;\mu)$ can be directly computed for each time slot according to \eqref{eq:ptstar}. 
Therefore, evaluating the power allocation over all $N$ time slots requires $\mathcal{O}(N)$ operations. 
The water level $\mu$ is determined from the total energy constraint by solving $S(\mu)=E_{\rm cap}$, where $S(\mu)=\sum_{n=1}^{N}P_t^\star(n;\mu)$ is a monotonic function. 
Thus, a bisection search can be used to find $\mu$ with accuracy $\epsilon_{\rm wf}$. 
Since each bisection step requires evaluating $S(\mu)$ over all $N$ time slots, and the number of bisection steps is $\mathcal{O}(\log(1/\epsilon_{\rm wf}))$, the complexity of the power-allocation step is
\begin{equation}
\mathcal{O}\left(N\log\left(\frac{1}{\epsilon_{\rm wf}}\right)\right).
\end{equation}

For fixed transmit powers, the UAV trajectory subproblem $(\mathcal{P}8)$ is solved using a standard convex solver. 
The optimization variables include the UAV 3D positions $\{\mathbf{q}_n\}_{n=1}^{N}$, velocities $\{V_n\}_{n=1}^{N}$, auxiliary variables $\{h_n\}_{n=1}^{N}$, and slack variables $\{t_n,u_n\}_{n=1}^{N}$. 
Thus, the number of scalar optimization variables scales linearly with $N$, i.e.,
\begin{equation}
n_v=\mathcal{O}(N).
\end{equation}
The number of constraints, including the mobility, propulsion-energy, SAR SNR, range/azimuth resolution, altitude, and slack-variable constraints, also scales linearly with $N$. 
Using an interior-point method, the complexity of solving $(\mathcal{P}8)$ can be approximately expressed as
\begin{equation}
\mathcal{O}\left(I_{\rm IP} n_v^3\right)
=
\mathcal{O}\left(I_{\rm IP}N^3\right).
\end{equation}

Consequently, the overall computational complexity of the proposed AO-based algorithm is approximately given by
\begin{equation}
\mathcal{O}\left(
I_{\rm AO}
\left[
N\log\left(\frac{1}{\epsilon_{\rm wf}}\right)
+
I_{\rm IP}N^{3}
\right]
\right).
\end{equation}
Since the trajectory update dominates the computational cost, the overall complexity is mainly governed by $\mathcal{O}(I_{\rm AO}I_{\rm IP}N^3)$.

Table~\ref{tab:complexity} summarizes the computational complexity of the main steps of the proposed algorithm. 
The above analysis also explains why simplified but physically meaningful models are adopted in this work. 
Specifically, the communication link is represented by an effective channel gain, and the SAR sensing requirement is characterized by SNR and spatial-resolution constraints. 
With these simplifications, the original HAPs--UAV ISARAC design can be formulated as a trajectory--power optimization problem over $N$ time slots. 
As a result, the transmit-power allocation admits a closed-form capped water-filling structure, and the UAV trajectory update can be solved through convex approximation with polynomial complexity. 
Therefore, the adopted modeling assumptions provide a tractable way to capture the main sensing--communication coupling while enabling efficient system-level trajectory and resource-allocation design.

\begin{table}[htp]
\centering
\caption{Computational complexity of the proposed AO-based algorithm}
\label{tab:complexity}
\begin{tabular}{c|c}
\hline
Step  & Complexity \\
\hline
Coefficient update 

& $\mathcal{O}(N)$ \\
\hline
Power allocation 

& $\mathcal{O}\left(N\log\left(\frac{1}{\epsilon_{\rm wf}}\right)\right)$ \\
\hline
Trajectory update 
&  $\mathcal{O}(I_{\rm IP}N^3)$ \\
\hline
Overall algorithm 
&  $\mathcal{O}\left(I_{\rm AO}\left[N\log\left(\frac{1}{\epsilon_{\rm wf}}\right)+I_{\rm IP}N^3\right]\right)$ \\
\hline
\end{tabular}
\end{table}

\section{Simulation Results}
\label{sec4}

This section presents simulation results to evaluate the proposed HAPs--UAV bistatic ISARAC system. 
Unless otherwise stated, the total sensing duration is set to 30 s with a sampling interval of 0.5 s. 
The HAPs operates at an altitude of 10 km, while the remaining parameter configurations are summarized in Table~\ref{tab2}. 
The simulation study is organized into two parts. 
First, we evaluate the convergence behavior, trajectory characteristics, power allocation, and sensing--communication tradeoff of the proposed HAPs--UAV bistatic ISARAC design. 
Second, we compare the proposed dual-platform architecture with a single-platform HAPs-based ISARAC benchmark to demonstrate the efficiency gain brought by UAV-assisted passive SAR echo reception.

\begin{table}[htp]
	\centering
	\caption{The parameters of the HAPs--UAV bistatic ISARAC \cite{hu2022trajectory,zeng2019energy}}
	\label{tab2}
	\begin{tabular}{c|c}
		\hline
		{Parameter} & {Value} 
		\\ \cline{1-2}	
		Transmit gain $G_t$ & 20 dBi
		\\
        Receive gain $G_r$ & 20 dBi 
		\\
        RCS $\sigma _0$ & 0.1  
        \\
        Loss factor $L_s$ & 3 dB
        \\
        Receiver temperature $T_b$ & 290 K
        \\
        Receiver noise figure $F_n$ & 4 dB 
        \\
        Bandwidth $B_r$ & 150 MHz
        \\
        CPI $T_w$ & 1 s
        \\
        Minimum range resolution $d_{\min}^{r}$ & 20 m
        \\
        Minimum azimuth resolution $d_{\min}^{a}$ & 20 m
        \\
        SAR observation angle $\alpha$ & $\pi/4$
        \\
        Maximum/minimum UAV height & 500 m, 10 m
        \\
        Maximum UAV speed & 100 m/s
        \\
        Blade profile power $P_0$ and tip speed $u_{\mathrm{tip}}$ & 3.4 W, 60 m/s
        \\
        Rotor induced power $P_1$ and velocity $v_0$ & 118 W, 5.4 m/s
        \\
        Rotor solidity and disc area, $s$ and $A$ & 0.02, 0.5 m$^2$
        \\
        Air density $\varrho$ and fuselage drag fraction $d_0$ & 1.225 kg/m$^3$, 0.3
        \\
        Energy available for ISARAC $E^s$ & 1.5 kW
        \\
        UAV energy $E$ & 10 kW
        \\
		\hline
	\end{tabular}
\end{table}

\subsection{Performance Evaluation of the Proposed HAPs--UAV Bistatic ISARAC System}
\label{subsec:performance_proposed}

It should be noted that, to date, no existing work has investigated three-dimensional trajectory planning for HAPs--UAV bistatic ISARAC systems under energy constraints. 
Therefore, in this subsection, we first evaluate the proposed algorithm and compare it with a communication-only HAPs system to illustrate the impact of the sensing constraints and joint trajectory--power optimization. 
A more direct comparison with a single-platform HAPs-based ISARAC benchmark is further provided in Section~\ref{subsec:single_haps_isarac}.

For the communication-only HAPs system, the corresponding optimization problem is formulated as
\begin{equation}
    \label{eq:P9}
    \begin{split}
\left( \mathcal{P} 9 \right) :&\underset{P_t(n)}{\max}\,\,\sum_{n=1}^N{R\left( n \right)}
\\
s.t.&C8:0\leqslant P_t(n)\leqslant P_{\max},\forall n,
\\
& C9:\sum_{n=1}^N{P_t(n)\delta}\leqslant E^{\mathrm{s}}.
    \end{split}
\end{equation}

For analytical tractability, we assume that in the communication-only scenario, the ground user is always located at the beam center of the HAPs transmission. A comparison between optimization problems $\left( \mathcal{P} 1 \right) $ and $\left( \mathcal{P} 9 \right) $ reveals that, in pursuit of maximizing the sum rate, $\left( \mathcal{P} 1 \right) $ not only adjusts the transmit power but also drives the ground user closer to the ground projection of the HAPs. Consequently, it is reasonable to anticipate that the sum rate performance of the communication-only system will be inferior to that of the proposed ISARAC-enabled system.

\begin{figure}[htp]
	\centering
	{
		{\includegraphics[width=0.38\textwidth]{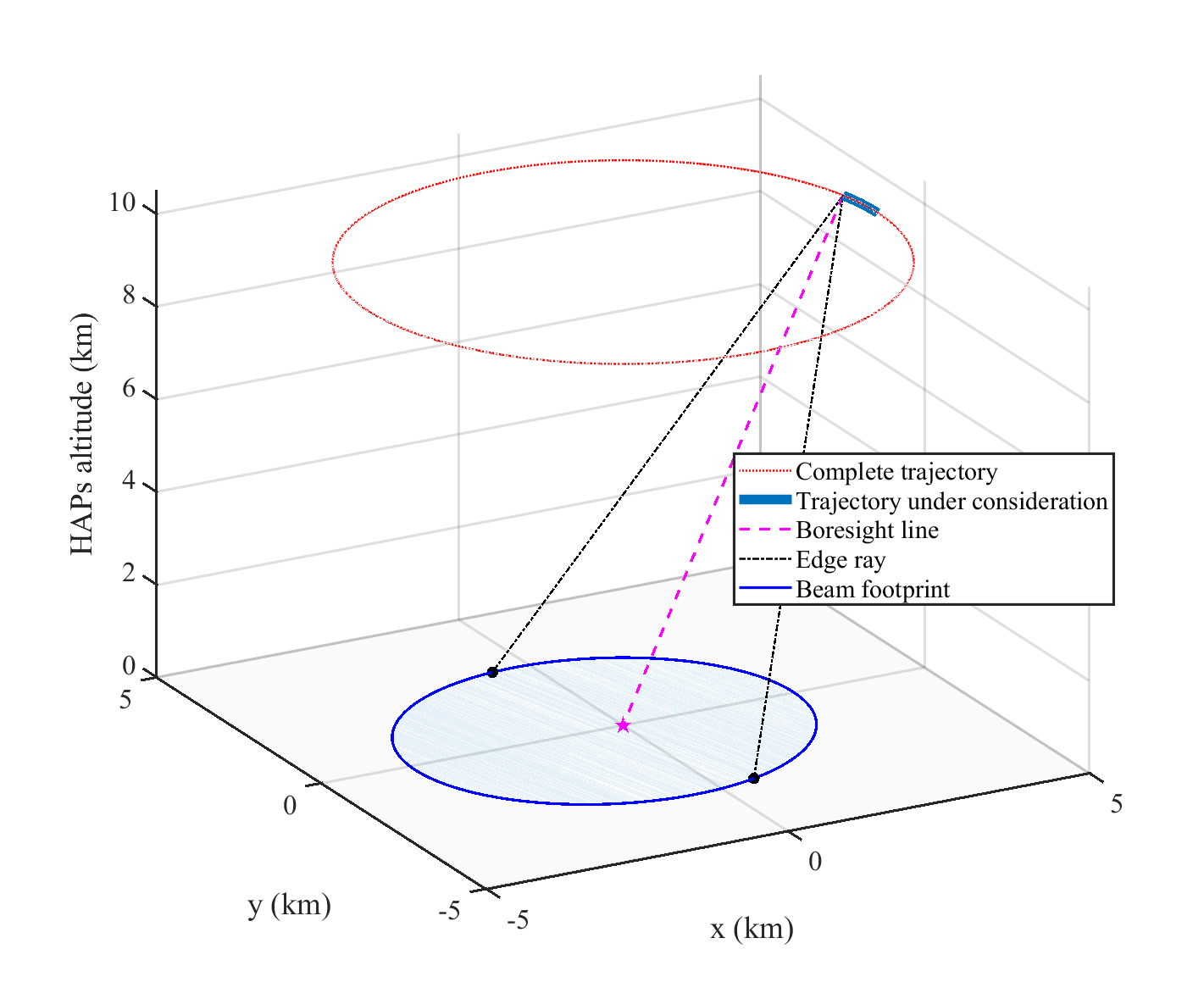}}}
	\caption{Illustration of the HAPs trajectory and the ground coverage of its signal beam.}
	\label{fig2}
\end{figure}
Fig. \ref{fig2} depicts the three-dimensional HAPs flight trajectory under the parameter settings considered in this section. The red circle indicates the complete flight path, while the blue segment highlights the portion corresponding to the sensing duration considered in this work. The shaded area on the ground represents the beam coverage region, within which users can establish communication links with the HAPs, while the signals reflected from targets can be received by the UAV to ultimately form SAR images.

\begin{figure}[htp]
	\centering
	{
		{\includegraphics[width=0.38\textwidth]{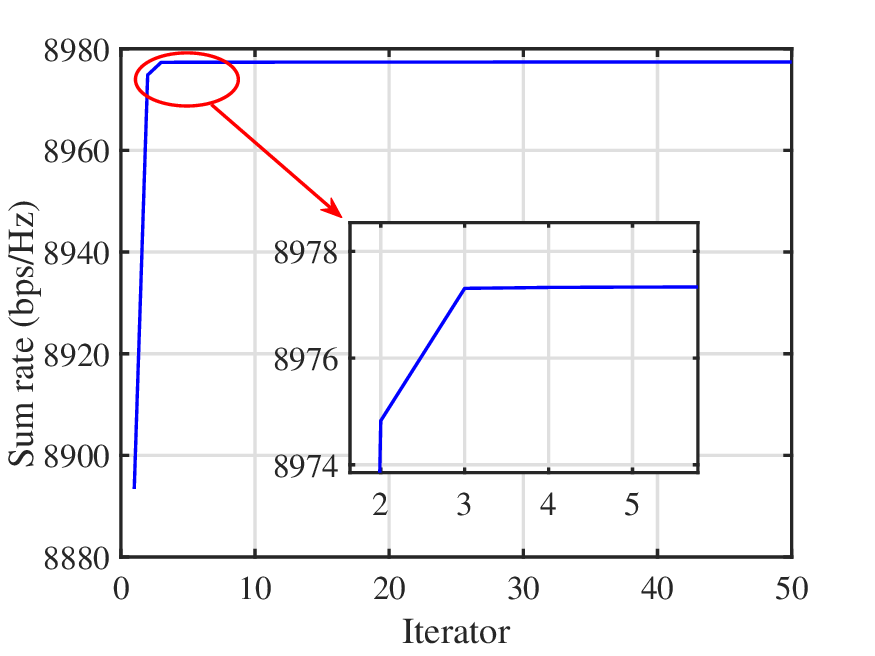}}}
	\caption{Convergence behavior of the proposed algorithm: objective function sum rate versus iteration index.}
	\label{fig3}
\end{figure}

Next, to validate the convergence of the proposed algorithm, Fig. \ref{fig3} shows the evolution of the objective function (sum rate) with respect to the iteration index. As observed, the sum rate monotonically increases and converges after approximately five iterations. This demonstrates that the proposed algorithm exhibits stable and fast convergence, thereby confirming its effectiveness in solving the formulated optimization problem.
\begin{figure}[htp]
	\centering
    \subfigure[]
	{
		{\includegraphics[width=0.43\textwidth]{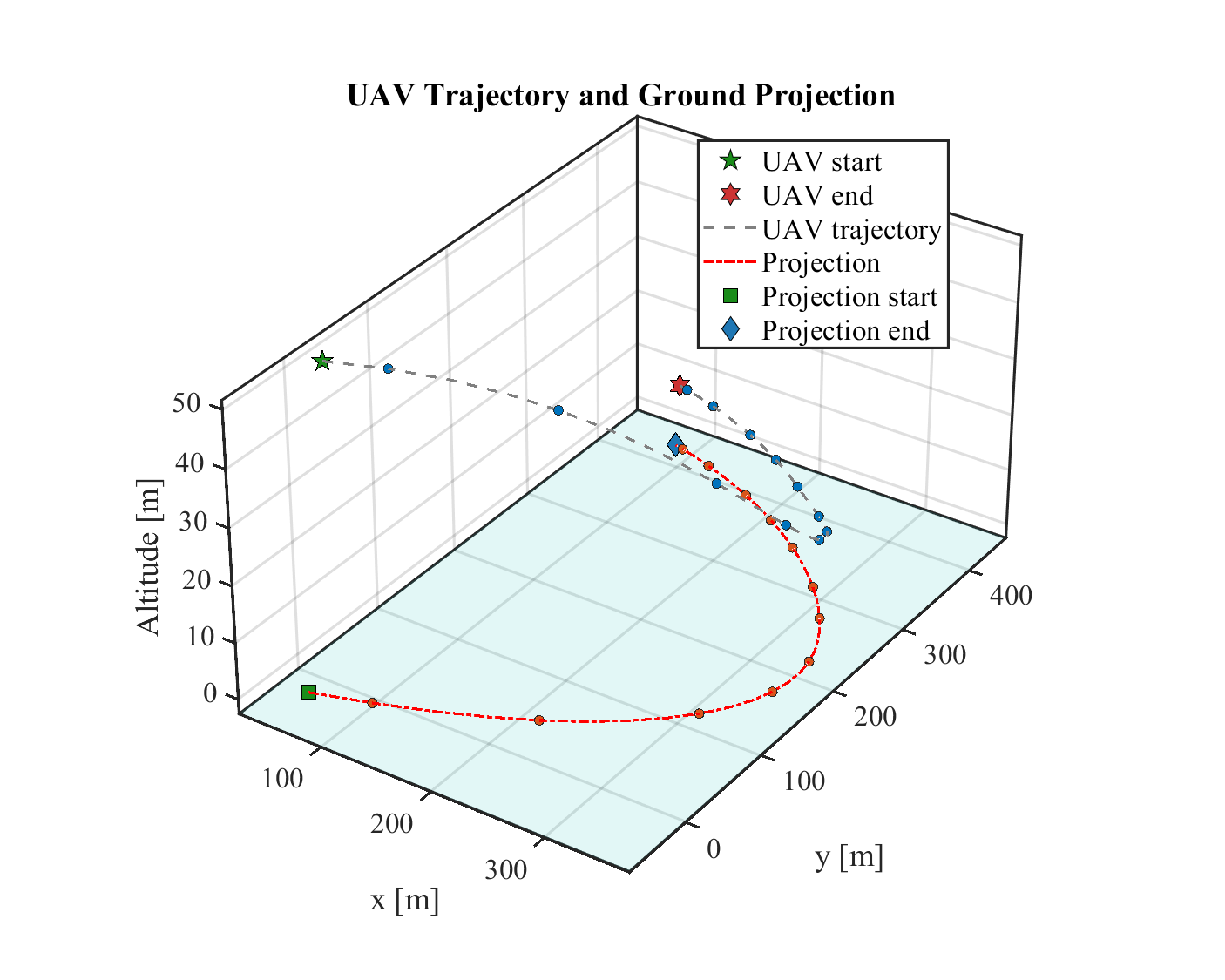}}}
	\subfigure[]{
		{\includegraphics[width=0.43\textwidth]{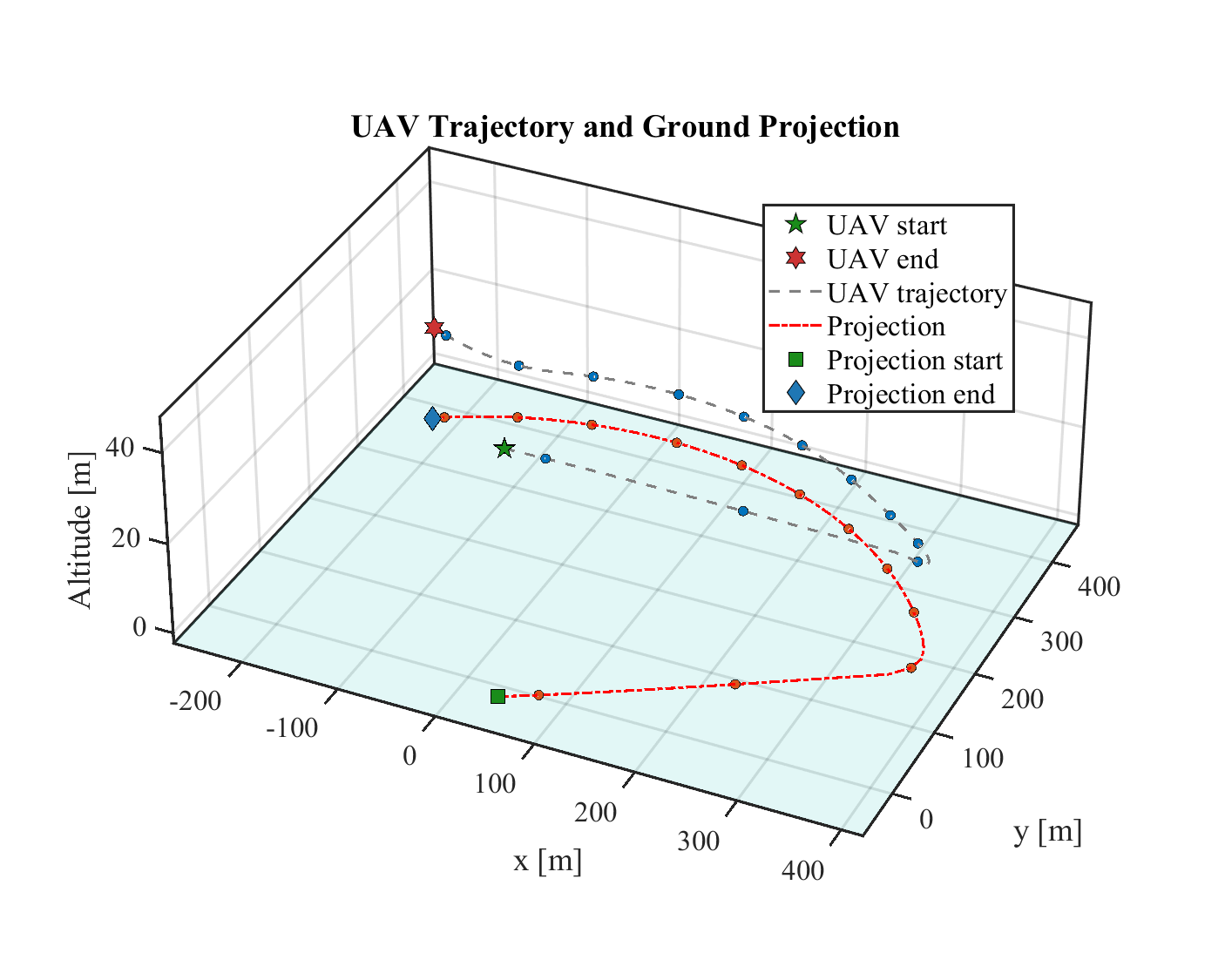}}}
	\subfigure[]{
		{\includegraphics[width=0.43\textwidth]{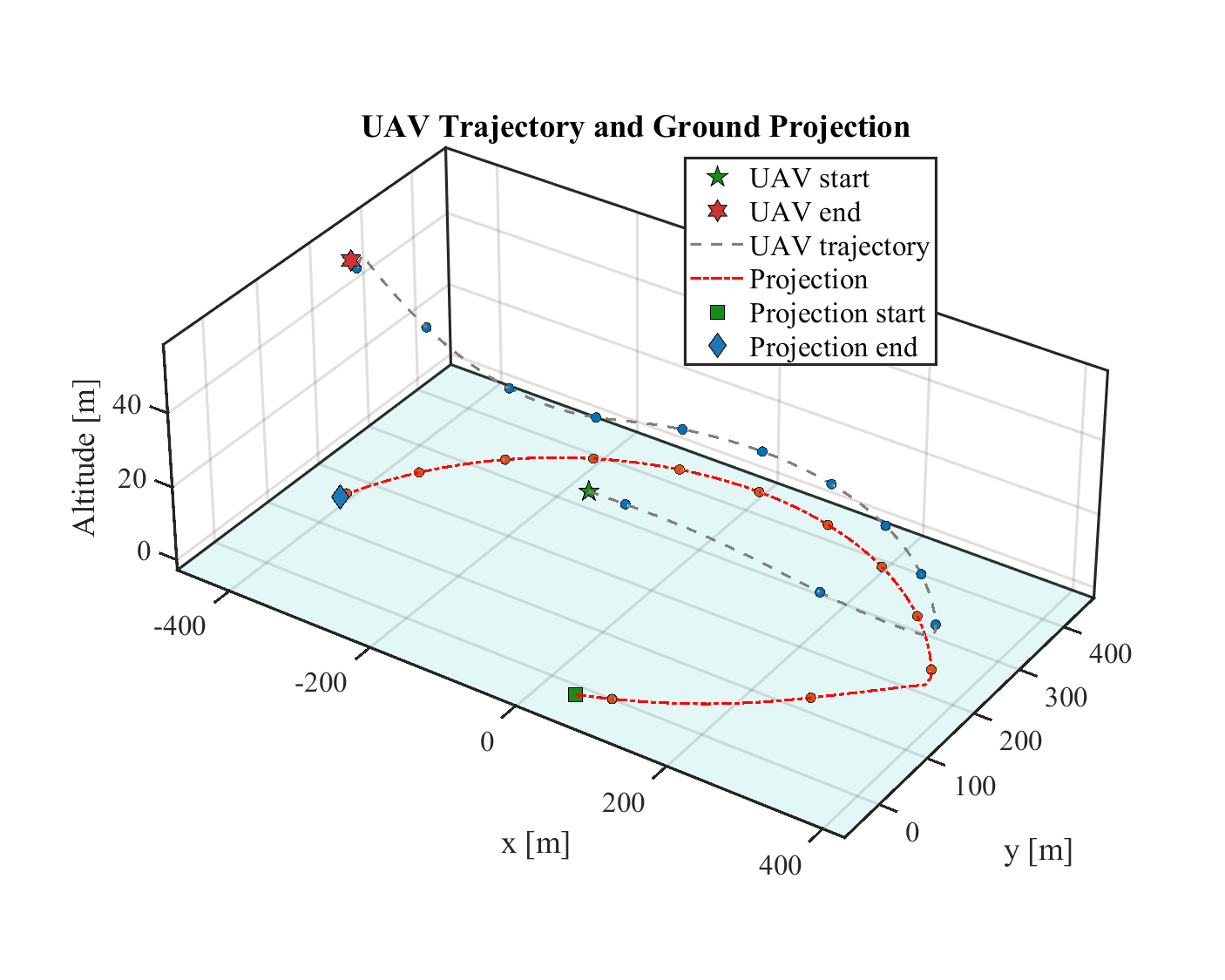}}}
	\caption{UAV 3D trajectories under different sensing durations and their ground projections. (a) $T=20s$, (b) $T=30s$, (c) $T=40s$.}
	\label{fig4}
\end{figure}

For the case of $T=20$s shown in Fig. \ref{fig4}(a), the dashed line indicates the UAV flight trajectory, while the solid red curve on the ground represents the projection of the UAV beam center. Although the projection appears approximately semicircular in shape, the actual trajectory exhibits variations in the altitude dimension. This clearly demonstrates the correctness of the three-dimensional trajectory planning considered in this paper. Moreover, according to the results of Fig.\ref{fig4}, we can know that as the sensing time increases from 20s to 40s, the UAV spans a larger horizontal range. Correspondingly, the ground projections (red curves) expand, indicating broader coverage areas. The start and end markers further confirm that longer sensing durations result in extended flight paths, thereby providing greater measurement diversity for SAR imaging and communication tasks. Moreover, the UAV trajectories corresponding to different sensing durations exhibit distinct shapes. This indicates that each trajectory is obtained through an overall optimization process, rather than by simply extending a previously optimized path. In other words, the trajectory design adapts to the specified sensing time in a holistic manner, leading to re-optimized flight paths that balance altitude variation, horizontal coverage, and communication–sensing performance trade-offs.

\begin{figure}[htp]
	\centering
	{
		{\includegraphics[width=0.38\textwidth]{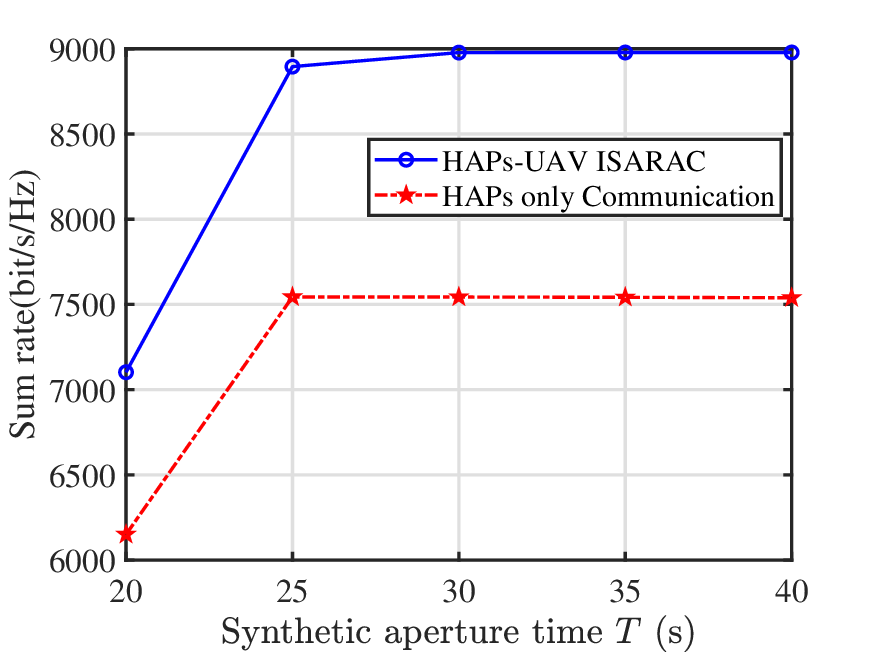}}}
	\caption{Performance comparison of the proposed ISARAC system and only communication system with different synthetic aperture time.}
	\label{fig5}
\end{figure}

Fig.\ref{fig5} presents a performance comparison between the proposed ISARAC system and the baseline system where the HAPs solely provides communication services. The results demonstrate that, under the specific simulation settings considered, the communication performance of the HAPs–UAV ISARAC system surpasses that of the communication-only case, which is consistent with the earlier analysis. Furthermore, it can be observed that the sum rate initially increases with the sensing duration. However, once the synthetic aperture time exceeds a certain threshold, further extensions yield negligible improvement. This saturation effect is attributed to constraint $C9$, which restricts the energy available for ISARAC operations on the HAPs, thereby capping the achievable sum rate.

\begin{figure}[htp]
	\centering
	{
		{\includegraphics[width=0.38\textwidth]{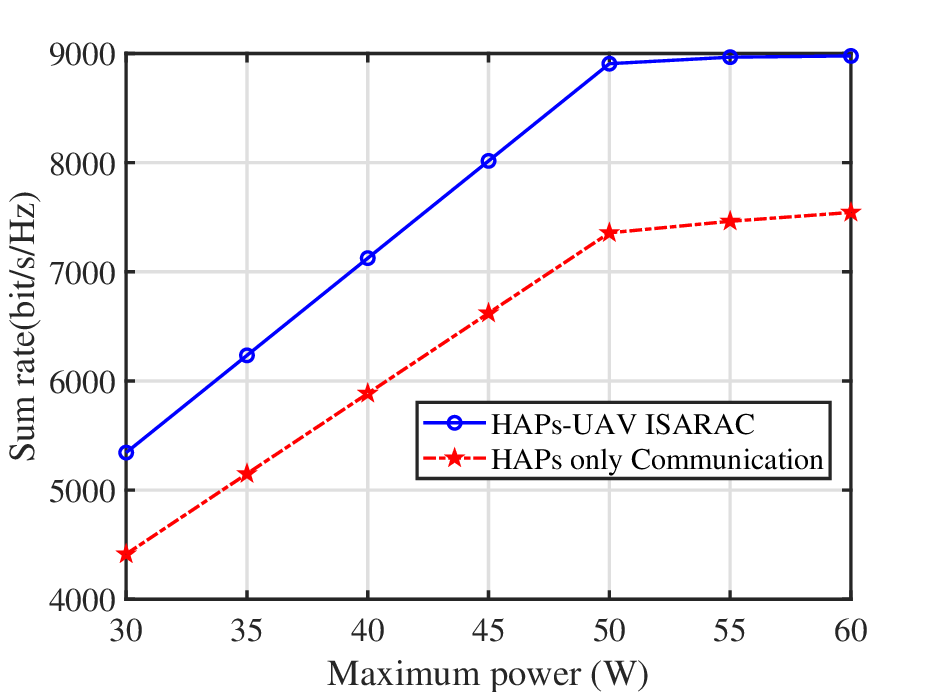}}}
	\caption{Performance comparison of the proposed ISARAC system and only communication system with different maximum transmit power.}
	\label{fig6}
\end{figure}

Fig.\ref{fig6} presents a performance comparison between the proposed ISARAC system and the communication-only system under different maximum transmit-power levels. It can be observed that as transmit power gradually increases, the sum rate of both systems improves, with the ISARAC system consistently maintaining a clear advantage across the entire range. 
In terms of trends, the sum rate grows rapidly with increasing $P_{\max}$ in the low-power region, primarily due to the significant improvement in the SNR. However, once the maximum transmit power exceeds 50 W, the performance gain begins to saturate, and further power increases yield only marginal improvements. This saturation effect arises mainly from two factors. First, the link quality approaches a near-saturation state at high power and second, additional resources are constrained by system-level limitations, such as total energy budget, preventing excess power from being efficiently converted into throughput gains.
These results suggest that simply increasing transmit power is not a sustainable strategy for performance improvement.



\begin{figure}[htp]
	\centering
	{
		{\includegraphics[width=0.38\textwidth]{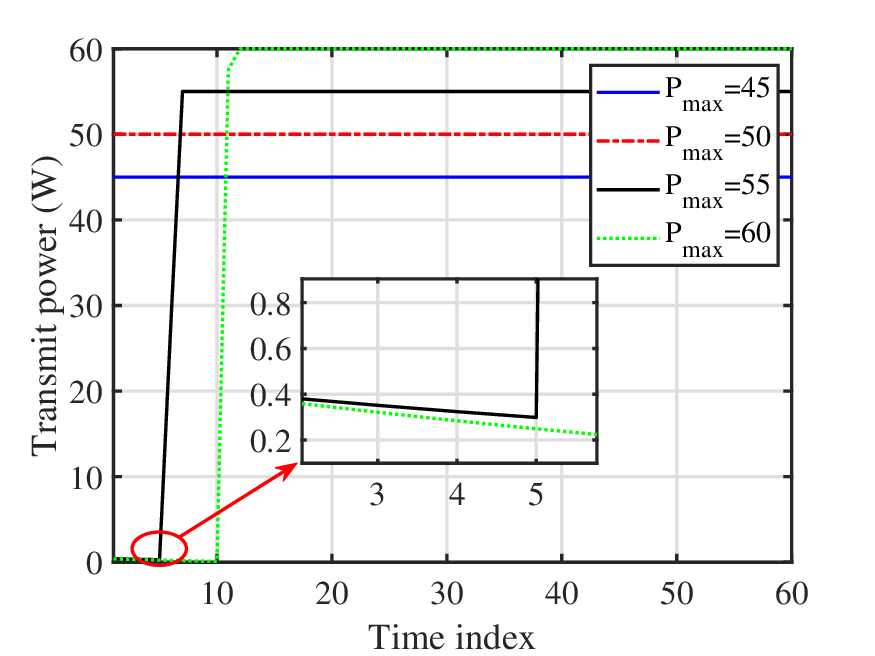}}}
	\caption{Transmit power vs time for varying $
P_{\max}
$.}
	\label{fig8}
\end{figure}

To further illustrate the role of constraints $C4$ and $C8$–$C9$ in transmit-power scheduling, Fig. \ref{fig8} presents the time evolution of the HAPs transmit power under different maximum transmit-power limits. As can be observed, when $P_{\max}$ is set to 45 W and 50 W, the transmit-power curve almost always adheres closely to the maximum power ceiling. This indicates that at lower maximum power levels, the transmit power is mainly governed by $C4$ (the SNR constraint) and $C8$ (the non-negativity and upper-bound condition on transmit power), with the system tending to fully exploit the available power in order to ensure link quality and satisfy the required SNR.
In contrast, when $P_{\max}$ increases to 55 W and 60 W, a significant difference arises. Although in principle higher instantaneous power could be provided, the total-energy budget constraint $C9$ prevents the HAPs from maintaining transmission at the peak power level throughout the mission duration. In other words, the system must dynamically regulate its transmit power over time to avoid prematurely depleting its limited energy resources. Under such conditions, transmit power is simultaneously constrained by $C4$, $C8$, and $C9$, with $C9$ becoming the key factor that determines whether the transmit power can be sustained near the peak level.
This phenomenon clearly reveals the interplay and trade-off among multiple constraints in power scheduling, namely at lower maximum power levels, the transmit-power profile is dominated by instantaneous constraints (SNR and peak-power bounds), whereas at higher maximum power levels, the system must strike a balance between achieving optimal instantaneous transmission quality and ensuring long-term energy sustainability. This highlights that in system design, both short-term performance and long-term energy endurance must be considered simultaneously in order to achieve robust and efficient HAPs mission operations.

\begin{figure}[htp]
	\centering
	{
		{\includegraphics[width=0.38\textwidth]{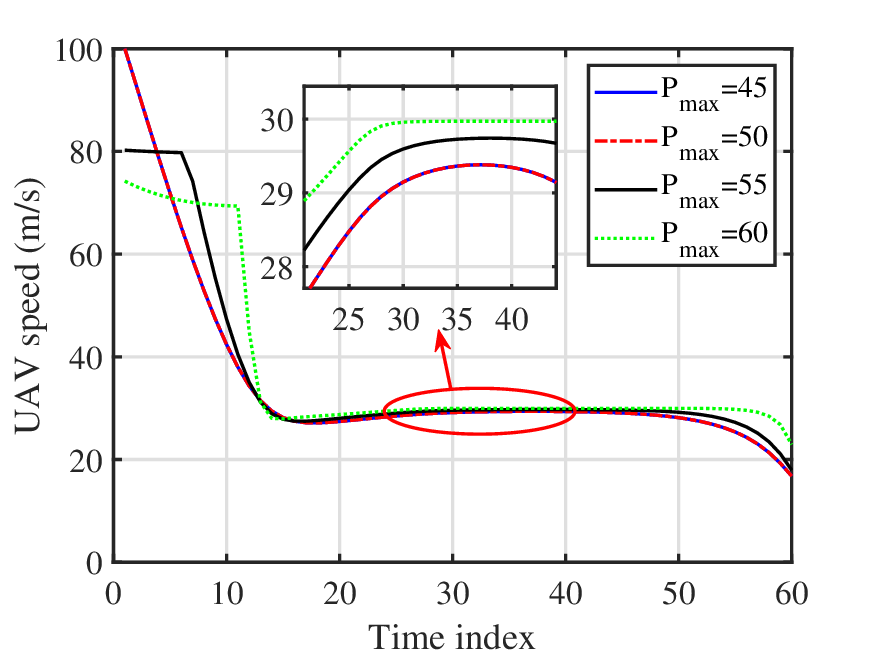}}}
	\caption{UAV speed vs time for varing $
P_{\max}
$.}
	\label{fig9}
\end{figure}

Fig. \ref{fig9} depicts the time evolution of UAV velocity under different maximum transmit-power constraints. It can be seen that variations in $P_{\max}$ exert only a marginal influence on the overall velocity profile, suggesting that transmit-power limits are not the dominant factor governing UAV motion in this scenario. Instead, the observed dynamics mainly reflect the intrinsic coupling between energy scheduling and UAV flight mechanics.
Across all cases, the UAV velocity exhibits a consistent decline–stabilization trend. In the initial phase, higher speed is maintained to support acceleration and attitude adjustments, followed by a gradual reduction and convergence toward a steady value as the flight enters a stable regime. This convergence indicates that the optimization framework effectively regulates velocity to balance propulsion energy with mission endurance.
Comparisons across different $P_{\max}$ further confirm that transmit-power constraints primarily induce minor numerical differences rather than altering the fundamental velocity pattern. Lower $P_{\max}$ values slightly reduce the UAV’s adjustment margin, while higher values yield small shifts in steady-state velocity, but the overall decline–stabilization trajectory remains unchanged.
This implies that, while transmit-power constraints strongly affect link quality and energy allocation, UAV velocity control is primarily shaped by mission design and long-term endurance considerations.

\subsection{Comparison With Single-Platform HAPs-Based ISARAC}
\label{subsec:single_haps_isarac}

To further evaluate the efficiency of the proposed dual-platform architecture, we introduce a single-platform HAPs-based ISARAC benchmark. 
In this benchmark, the HAPs transmits ISARAC waveforms to serve ground users and simultaneously receives the ground-reflected echoes for SAR sensing, without UAV-assisted passive reception. 
The HAPs follows the same circular trajectory and is subject to the same transmit-power and energy constraints as in the proposed scheme. 
Therefore, this benchmark provides a direct comparison between the proposed HAPs--UAV bistatic architecture and a single-platform ISARAC system under the same sensing and communication requirements.

Let $\mathbf{s}_n=[x_{s,n},y_{s,n},0]^T$ denote the ground scene center illuminated by the HAPs at time slot $n$, and define $d_{{\rm H},n}=\|\boldsymbol{p}_n-\mathbf{s}_n\|$. 
For the single-platform HAPs-based SAR sensing link, the received echo SNR can be expressed as
\begin{equation}
\mathrm{SNR}_{\rm H}(n)
=
\frac{P_t(n)G_tG_{\rm H}^{r}\lambda ^3\sigma _0A_p}
{\left(4\pi\right)^3d_{{\rm H},n}^{4}KT_bF_nB_rL_s},
\end{equation}
where $G_{\rm H}^{r}$ denotes the HAPs receive gain for SAR echo reception. 
For fairness, we set $G_{\rm H}^{r}=G_r$ in the simulations. 
The SNR requirement $\mathrm{SNR}_{\rm H}(n)\ge \mathrm{SNR}_{\min}$ yields the following sensing-induced transmit-power lower bound, namely
\begin{equation}
P_t(n)\ge L_{{\rm H},n},
\end{equation}
where
\begin{equation}
L_{{\rm H},n}
=
\frac{
\mathrm{SNR}_{\min}
\left(4\pi\right)^3
d_{{\rm H},n}^{4}
KT_bF_nB_rL_s
}
{
G_tG_{\rm H}^{r}\lambda ^3\sigma_0A_p
}.
\end{equation}

Accordingly, the single-platform HAPs-based ISARAC benchmark can be formulated as a sensing-constrained power-allocation problem. 
As detailed in Appendix~\ref{app:haps_isarac_baseline}, its optimal solution also admits a capped water-filling form, as
\begin{equation}
\label{eq60}
P_{{\rm H},t}^{\star}(n;\mu)
=
\left[
\mu-\frac{1}{a_{{\rm H},n}}
\right]_{L_{{\rm H},n}}^{P_{\max}},
\quad n=1,\ldots,N,
\end{equation}
where
\begin{equation}
a_{{\rm H},n}
=
\frac{G_tG_r^c\rho_0}
{d_{{\rm H},n}^{2}\tilde{\sigma}_0^2},
\end{equation}
and the water level $\mu$ is determined by the total energy constraint. 
If $\sum_{n=1}^{N}L_{{\rm H},n}\delta>E^{\rm s}$, the single-platform HAPs-based ISARAC benchmark is infeasible under the considered sensing requirement.

\begin{figure}[htp]
    \centering
    \subfigure[]{
        \includegraphics[width=0.43\textwidth]{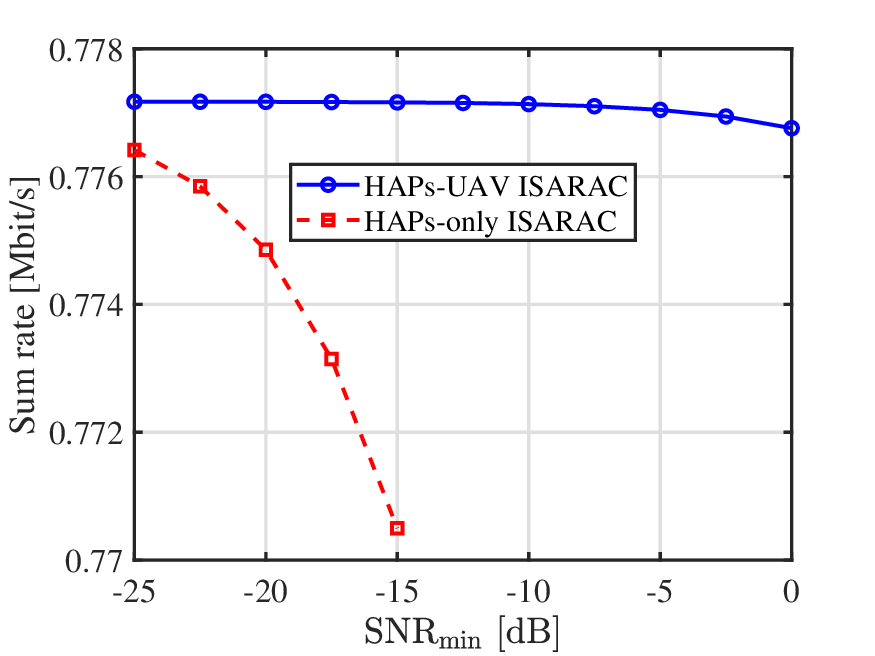}
    }
    \subfigure[]{
        \includegraphics[width=0.43\textwidth]{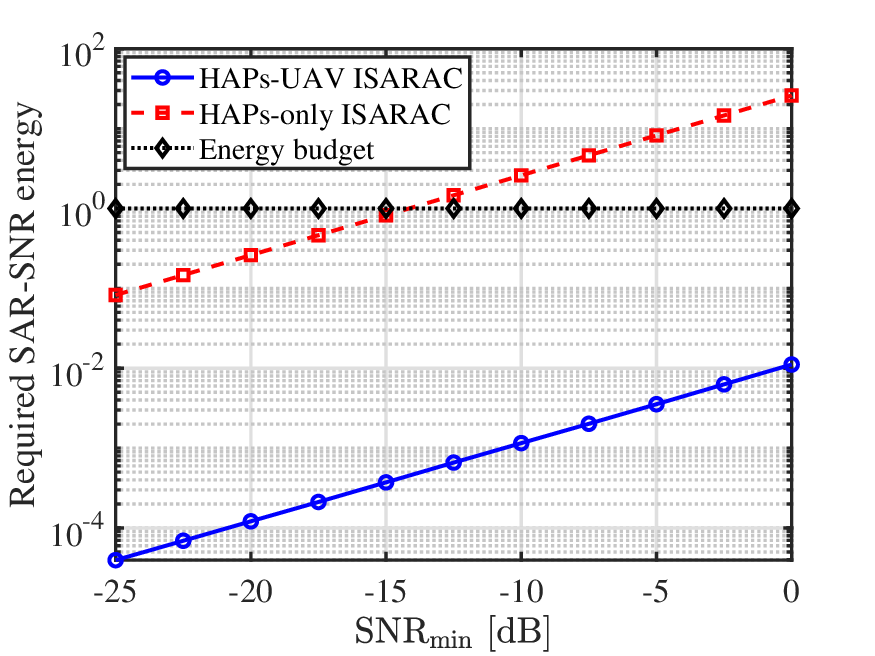}
    }
\caption{Sensing-constrained performance comparison between the proposed HAPs--UAV bistatic ISARAC architecture and the single-platform HAPs-based ISARAC benchmark under different SAR SNR thresholds. 
(a) Achievable sum rate under the SAR SNR constraint. 
(b) Normalized minimum transmit-energy floor required to satisfy the SAR SNR constraint, where values above one indicate infeasibility under the available ISARAC energy budget.}
    \label{fig:single_haps_compare}
\end{figure}

Fig.~\ref{fig:single_haps_compare} compares the proposed HAPs--UAV bistatic ISARAC architecture with the single-platform HAPs-based ISARAC benchmark under different SAR SNR thresholds. 
For this comparison, the SAR SNR threshold $\mathrm{SNR}_{\min}$ is varied over a prescribed range. 
For each value of $\mathrm{SNR}_{\min}$, the proposed AO-based algorithm is re-executed to jointly optimize the UAV trajectory and HAPs transmit-power allocation. 
After convergence, the achieved communication sum rate is recorded for the proposed HAPs--UAV bistatic ISARAC architecture. 
For the single-platform HAPs-based ISARAC benchmark, the HAPs trajectory and the illuminated ground scene are kept the same, and the corresponding sensing-constrained power-allocation problem is solved using the capped water-filling solution \eqref{eq60}. 
In Fig.~\ref{fig:single_haps_compare}(b), the plotted energy floor is defined as the minimum transmit energy required to satisfy the SAR SNR constraint, i.e., $\sum_{n=1}^{N}L_n\delta$, normalized by the available ISARAC energy budget $E^{\mathrm{s}}$, where $L_n$ denotes the sensing-induced lower bound on the transmit power. 
Therefore, a value larger than one indicates that the SAR SNR requirement cannot be satisfied under the given energy budget.

As shown in Fig.~\ref{fig:single_haps_compare}(a), the proposed HAPs--UAV bistatic ISARAC architecture maintains an almost stable sensing-constrained sum rate over a wide range of $\mathrm{SNR}_{\min}$ thresholds. 
In contrast, the single-platform HAPs-based ISARAC benchmark experiences a clear rate degradation as the SAR SNR requirement becomes more stringent and eventually becomes infeasible when the required sensing energy exceeds the available ISARAC energy budget. 
The reason for this behavior is further illustrated in Fig.~\ref{fig:single_haps_compare}(b). 
For the single-platform benchmark, the SAR echo experiences a long HAPs--ground--HAPs propagation path, leading to a rapidly increasing sensing-induced transmit-energy floor as the $\mathrm{SNR}_{\min}$ increases. 
By contrast, in the proposed HAPs--UAV bistatic architecture, the UAV acts as a low-altitude passive receiver and substantially shortens the target-to-receiver propagation path. 
As a result, the normalized minimum transmit energy required to satisfy the $\mathrm{SNR}_{\min}$ constraint remains far below the available energy budget, which explains why the communication sum rate of the proposed architecture is much less sensitive to the sensing threshold. 
Therefore, compared with the single-platform HAPs-based ISARAC benchmark, the proposed dual-platform design provides a wider feasible operating region and more efficiently balances persistent communication and SAR sensing requirements.

\section{Conclusion}
\label{sec5}
This paper investigates UAV 3D trajectory planning and resource allocation for a HAPs–UAV bistatic ISARAC architecture under imaging constraints, a total-energy constraint, and UAV dynamics constraints. Within this architecture, the HAPs transmits the ISARAC waveform, while the UAV collects ground-target echoes for SAR imaging and ground users access communication links via the HAPs, thereby realizing an ISARAC framework. To maximize the ground users’ sum rate, we formulated a joint optimization problem and addressed its inherent nonconvexity by developing AO framework. Specifically, a closed-form water-filling solution was derived for the power-allocation subproblem with fixed UAV states, while the UAV trajectory subproblem was reformulated through SCA/DC techniques into a tractable convex program. Simulation results confirmed the effectiveness of the proposed approach. 
Further comparison with a single-platform HAPs-based ISARAC benchmark showed that the proposed HAPs--UAV bistatic architecture significantly reduces the sensing-induced energy floor, thereby maintaining a more stable communication rate and a wider feasible operating region under stringent SAR SNR requirements.
Under stringent imaging thresholds and limited energy budgets, the algorithm achieved more favorable energy allocation and trajectory geometries between sensing and communications compared with baseline strategies, improving the sensing-constrained communication performance. Moreover, the results corroborated the theoretical mechanism that tightening the imaging SNR threshold reduces the degrees of freedom for communication-oriented power allocation, thereby causing the sum rate to monotonically decrease and approach a lower performance floor. Overall, this work provides a resource--trajectory design framework for low-altitude ISARAC platforms. 
 The practical implications of the adopted model simplifications are also noted. 
The HAPs energy model constrains only the energy available for ISARAC transmission, while practical payload energy is jointly affected by solar energy harvesting, battery charging/discharging, platform motion and station-keeping energy consumption, and transmit-power demand. 
The LoS-dominated free-space communication channel may also be affected by atmospheric attenuation, rain attenuation, shadowing, and small-scale fading, which can reduce the effective channel gain and achievable rate. 
Moreover, the adopted SNR and spatial-resolution constraints provide tractable SAR sensing metrics, but practical image quality further depends on the point spread function, peak sidelobe ratio, integrated sidelobe ratio, and motion-induced phase errors caused by UAV positioning and trajectory uncertainties. 
These factors may require more conservative power allocation and trajectory planning. 
Nevertheless, they mainly affect the available energy margin, link budget, imaging-quality margin, and feasible operating region, while the main sensing--communication tradeoff and optimization structure remain applicable. 
Future work will incorporate solar-aware HAPs energy dynamics, environment-dependent channel models, SAR image-quality metrics, robust trajectory design under UAV positioning errors, multi-UAV cooperation, and airspace-aware trajectory optimization.

\appendices
\section{Derivation of Eq.\eqref{eqq28}}
\label{refA}
In the first, the expression of $
\left\| \boldsymbol{p}_n-\mathbf{q}_{n}^{c} \right\| ^2
$ can be rewritten as 
\begin{equation}
    \label{}
    \left\| \boldsymbol{p}_n-\mathbf{q}_{n}^{c} \right\| ^2=(x_{p}^{n}-x_{n}^{c})^2+(y_{p}^{n}-y_{n}^{c})^2+H^2
\end{equation}

Accordingly, the gradient of $
g_{5}^{2}\left( \mathbf{q}_n \right) 
$ with respect to $\mathbf{q}_n$ can be expressed as


\begin{equation}
    \label{}
\nabla _{\mathbf{q}_n}g_{5}^{2}\left( \mathbf{q}_n \right) =\left[ \begin{matrix}
	\frac{\partial \left\| \boldsymbol{p}_n-\mathbf{q}_{n}^{c} \right\| ^2}{\partial x_n}&		\frac{\partial \left\| \boldsymbol{p}_n-\mathbf{q}_{n}^{c} \right\| ^2}{\partial y_n}&		\frac{\partial \left\| \boldsymbol{p}_n-\mathbf{q}_{n}^{c} \right\| ^2}{\partial z_n}\\
\end{matrix} \right] ^T
\end{equation}
with $
\frac{\partial g_{5}^{2}\left( \mathbf{q}_n \right) }{\partial x_n}
$ denotes the  partial derivative of $g_{5}^{2}\left( \mathbf{q}_n \right) $ with respect to $x_n$, namely 
\begin{equation}
    \label{}
    \begin{split}
\frac{\partial g_{5}^{2}\left( \mathbf{q}_n \right) }{\partial x_n}=&2(x_{n}^{c}-x_{p}^{n})\cdot \frac{\partial x_{n}^{c}}{\partial x_n}+2(y_{n}^{c}-y_{p}^{n})\cdot \frac{\partial y_{n}^{c}}{\partial x_n}
\\
=&2(x_{n}^{c}-x_{p}^{n})
    \end{split}
\end{equation}
Similarly, the partial derivative with respect to $y_n$ can be computed as 
\begin{equation}
    \label{}
    \begin{split}
\frac{\partial g_{5}^{2}\left( \mathbf{q}_n \right) }{\partial y_n}=&2(x_{n}^{c}-x_{p}^{n})\cdot \frac{\partial x_{n}^{c}}{\partial y_n}+2(y_{n}^{c}-y_{p}^{n})\cdot \frac{\partial y_{n}^{c}}{\partial y_n}
\\
=&2(y_{n}^{c}-y_{p}^{n})
    \end{split}
\end{equation}

Finally, the partial derivative with respect to $z_n$ yields
\begin{equation}
    \label{}
    \begin{split}
\frac{\partial g_{5}^{2}\left( \mathbf{q}_n \right) }{\partial z_n}=&2(x_{n}^{c}-x_{p}^{n})\cdot \frac{\partial x_{n}^{c}}{\partial z_n}+2(y_{n}^{c}-y_{p}^{n})\cdot \frac{\partial y_{n}^{c}}{\partial z_n}
\\
=&2\tan \eta \left[ (x_{n}^{c}-x_{p}^{n})\sin \gamma _n-(y_{n}^{c}-y_{p}^{n})\cos \gamma _n \right] 
    \end{split}
\end{equation}
By combining all the above results, we obtain the results \eqref{eqq28}.

\section{Single-Platform HAPs-Based ISARAC Benchmark}
\label{app:haps_isarac_baseline}

This appendix provides the formulation and solution of the single-platform HAPs-based ISARAC benchmark used for comparison in Section~\ref{subsec:single_haps_isarac}. 
In this benchmark, the HAPs simultaneously supports downlink communication and SAR sensing without UAV-assisted passive reception.

The HAPs--ground communication channel gain is given by
\begin{equation}
\rho_{{\rm H},n}=\frac{\rho_0}{d_{{\rm H},n}^{2}},
\end{equation}
and the corresponding achievable rate is
\begin{equation}
R_{\rm H}(n)
=
B\log\left(1+
\frac{P_t(n)G_tG_r^c\rho_{{\rm H},n}}
{\tilde{\sigma}_0^2}
\right).
\end{equation}
Define
\begin{equation}
a_{{\rm H},n}
=
\frac{G_tG_r^c\rho_0}
{d_{{\rm H},n}^{2}\tilde{\sigma}_0^2}.
\end{equation}
Then the achievable rate can be rewritten as
\begin{equation}
R_{\rm H}(n)=B\log\left(1+a_{{\rm H},n}P_t(n)\right).
\end{equation}

With the sensing-induced lower bound $L_{{\rm H},n}$ defined in Section~\ref{subsec:single_haps_isarac}, the single-platform HAPs-based ISARAC benchmark is formulated as
\begin{align}
\left(\mathcal{P}_{\rm H}\right):
\underset{\{P_t(n)\}}{\max}\quad
&\sum_{n=1}^{N}B\log\left(1+a_{{\rm H},n}P_t(n)\right)
\\
{\rm s.t.}\quad
&L_{{\rm H},n}\le P_t(n)\le P_{\max},\quad \forall n,
\\
&\sum_{n=1}^{N}P_t(n)\delta\le E^{\rm s}.
\end{align}
This problem is convex because the objective function is concave in $\{P_t(n)\}$ and all constraints are affine. 
Following the Proposition \ref{prop:cwf} again, the optimal transmit power is obtained as
\begin{equation}
P_{{\rm H},t}^{\star}(n;\mu)
=
\left[
\mu-\frac{1}{a_{{\rm H},n}}
\right]_{L_{{\rm H},n}}^{P_{\max}},
\quad n=1,\ldots,N,
\end{equation}
where $\mu$ is the water level. 
Define
\begin{equation}
S_{\rm H}(\mu)=
\sum_{n=1}^{N}
\left[
\mu-\frac{1}{a_{{\rm H},n}}
\right]_{L_{{\rm H},n}}^{P_{\max}}.
\end{equation}
If
\begin{equation}
\sum_{n=1}^{N}L_{{\rm H},n}\delta
\le
E^{\rm s}
\le
\sum_{n=1}^{N}P_{\max}\delta,
\end{equation}
there exists a water level $\mu$ such that
\begin{equation}
S_{\rm H}(\mu)=\frac{E^{\rm s}}{\delta}.
\end{equation}
If $E^{\rm s}\ge \sum_{n=1}^{N}P_{\max}\delta$, then $P_{{\rm H},t}^{\star}(n)=P_{\max}$ for all $n$. 
If $\sum_{n=1}^{N}L_{{\rm H},n}\delta>E^{\rm s}$, the benchmark is infeasible under the considered sensing requirement.

\bibliographystyle{IEEEtran}
\bibliography{ref}

@book{boyd2004convex,
  title={Convex optimization},
  author={Boyd, Stephen P and Vandenberghe, Lieven},
  year={2004},
  publisher={Cambridge university press}
}

@article{rodriguez2009bistatic,
  title={Bistatic {TerraSAR-X/F-SAR} spaceborne--airborne {SAR} experiment: description, data processing, and results},
  author={Rodriguez-Cassola, Marc and Baumgartner, Stefan V and Krieger, Gerhard and Moreira, Alberto},
  journal={IEEE Transactions on Geoscience and Remote Sensing},
  volume={48},
  number={2},
  pages={781--794},
  year={2009},
  publisher={IEEE}
}

@article{kawamoto2023hapsbased,
  title={HAPS-based interference suppression through null broadening with directivity control in space-air-ground integrated networks},
  author={Kawamoto, Yuichi and Matsushita, Akinori and Verma, Shikhar and Kato, Nei and Kaneko, Kazuma and Sata, Ayaka and Hangai, Masatake},
  journal={IEEE Transactions on Vehicular Technology},
  volume={72},
  number={12},
  pages={16098--16107},
  year={2023},
  publisher={IEEE}
}

@article{li2024uav,
  title={{UAV-RIS-aided} space-air-ground integrated network: Interference alignment design and DoF analysis},
  author={Li, Jingfu and Chen, Gaojie and Zhang, Tong and Feng, Wenjiang and Jiang, Weiheng and Quek, Tony QS and Tafazolli, Rahim},
  journal={IEEE Transactions on Wireless Communications},
  volume={23},
  number={9},
  pages={11678--11692},
  year={2024},
  publisher={IEEE}
}

@article{kawamoto2024interference,
  title={Interference suppression in {HAPS-based} space-air-ground integrated networks using a codebook-based approach},
  author={Kawamoto, Yuichi and Okawara, Yuto and Verma, Shikhar and Kato, Nei and Kaneko, Kazuma and Sata, Ayaka and Ochiai, Mari},
  journal={IEEE Transactions on Vehicular Technology},
  year={2024},
  publisher={IEEE}
}

@inproceedings{rodrigues2025weather,
  title={Weather Attenuation Dataset Generation Method for Prediction-based Control of Non-Terrestrial High-Frequency Wireless Networks},
  author={Rodrigues, Tiago Koketsu and Songsriboonsit, Norranat and Verma, Shikhar},
  booktitle={2025 IEEE VTS Asia Pacific Wireless Communications Symposium (APWCS)},
  pages={1--5},
  year={2025},
  organization={IEEE}
}

@article{zedini2024improving,
  title={Improving performance of integrated ground-{HAPS} FSO communication links with {MIMO} application},
  author={Zedini, Emna and Ata, Yal{\c{c}}{\i}n and Alouini, Mohamed-Slim},
  journal={IEEE Photonics Journal},
  volume={16},
  number={2},
  pages={1--14},
  year={2024},
  publisher={IEEE}
}

@article{dahrouj2023machine,
  title={Machine learning-based user scheduling in integrated satellite-HAPS-ground networks},
  author={Dahrouj, Hayssam and Liu, Shasha and Alouini, Mohamed-Slim},
  journal={IEEE Network},
  volume={37},
  number={2},
  pages={102--109},
  year={2023},
  publisher={IEEE}
}

@article{kawamoto2023traffic,
  title={Traffic-prediction-based dynamic resource control strategy in {HAPS-mounted MEC-assisted} satellite communication systems},
  author={Kawamoto, Yuichi and Takahashi, Masaki and Verma, Shikhar and Kato, Nei and Tsuji, Hiroyuki and Miura, Amane},
  journal={IEEE Internet of Things Journal},
  volume={11},
  number={8},
  pages={13824--13836},
  year={2023},
  publisher={IEEE}
}

@article{rockafellar2015convex,
  title={Convex analysis:(pms-28)},
  author={Rockafellar, Ralph Tyrell},
  year={2015},
  publisher={Princeton university press}
}

@article{cumming2005digital,
  title={Digital processing of synthetic aperture radar data},
  author={Cumming, Ian G and Wong, Frank H},
  journal={Artech house},
  volume={1},
  number={3},
  pages={108--110},
  year={2005},
  publisher={Boston}
}

@article{palomar2006tutorial,
  title={A tutorial on decomposition methods for network utility maximization},
  author={Palomar, Daniel P{\'e}rez and Chiang, Mung},
  journal={IEEE Journal on Selected Areas in Communications},
  volume={24},
  number={8},
  pages={1439--1451},
  year={2006},
  publisher={IEEE}
}

@book{hein2003processing,
  title={Processing of SAR data},
  author={Hein, Achim},
  year={2003},
  publisher={Springer}
}

@article{huang2025joint,
  title={Joint Trajectory and Resource Optimization for {HAPs-SAR} Systems with Energy-Aware Constraints},
  author={Huang, Bang and Park, Kihong and Pang, Xiaowei and Alouini, Mohamed-Slim},
  journal={arXiv preprint arXiv:2506.23248},
  year={2025}
}

@article{moccia2011spatial,
  title={Spatial resolution of bistatic synthetic aperture radar: Impact of acquisition geometry on imaging performance},
  author={Moccia, Antonio and Renga, Alfredo},
  journal={IEEE Transactions on Geoscience and Remote Sensing},
  volume={49},
  number={10},
  pages={3487--3503},
  year={2011},
  publisher={IEEE}
}

@article{dower2018bistatic,
  title={Bistatic {SAR}: Forecasting spatial resolution},
  author={Dower, William and Yeary, Mark},
  journal={IEEE Transactions on Aerospace and Electronic Systems},
  volume={55},
  number={4},
  pages={1584--1595},
  year={2018},
  publisher={IEEE}
}

@article{zeng2019energy,
  title={Energy minimization for wireless communication with rotary-wing {UAV}},
  author={Zeng, Yong and Xu, Jie and Zhang, Rui},
  journal={IEEE transactions on wireless communications},
  volume={18},
  number={4},
  pages={2329--2345},
  year={2019},
  publisher={IEEE}
}

@inproceedings{jirousek2023design,
  title={Design of a Synthetic Aperture Radar Instrument for a High-Altitude Platform},
  author={Jirousek, Matthias and Peichl, Markus and Anger, Simon and Dill, Stephan and Engel, Marius},
  booktitle={IGARSS 2023-2023 IEEE International Geoscience and Remote Sensing Symposium},
  pages={2045--2048},
  year={2023},
  organization={IEEE}
}

@article{belmekkiCellularNetworkSky2024,
	title = {Cellular {{Network From}} the {{Sky}}: {{Toward People-Centered Smart Communities}}},
	shorttitle = {Cellular {{Network From}} the {{Sky}}},
	author = {Belmekki, Baha Eddine Youcef and Aljohani, Abdulah Jeza and Althubaity, Saud A. and Harthi, Abdulhadi Al and Bean, Kevin and Aijaz, Adnan and Alouini, Mohamed-Slim},
	date = {2024},
	journal = {IEEE Open Journal of the Communications Society},
	volume = {5},
	pages = {1916--1936},
	issn = {2644-125X},
	doi = {10.1109/OJCOMS.2024.3378297},
	langid = {english}
}

@article{lyu2022joint,
  title={Joint maneuver and beamforming design for {UAV-enabled} integrated sensing and communication},
  author={Lyu, Zhonghao and Zhu, Guangxu and Xu, Jie},
  journal={IEEE Transactions on Wireless Communications},
  volume={22},
  number={4},
  pages={2424--2440},
  year={2022},
  publisher={IEEE}
}

@article{liu2024uav,
  title={{UAV} assisted integrated sensing and communications for internet of things: {3D} trajectory optimization and resource allocation},
  author={Liu, Zechen and Liu, Xin and Liu, Yuemin and Leung, Victor CM and Durrani, Tariq S},
  journal={IEEE Transactions on Wireless Communications},
  volume={23},
  number={8},
  pages={8654--8667},
  year={2024},
  publisher={IEEE}
}

@article{liu2024secure,
  title={Secure rate maximization for {ISAC-UAV} assisted communication amidst multiple eavesdroppers},
  author={Liu, Yuemin and Liu, Xin and Liu, Zechen and Yu, Yingfeng and Jia, Min and Na, Zhenyu and Durrani, Tariq S},
  journal={IEEE Transactions on Vehicular Technology},
  volume={73},
  number={10},
  pages={15843--15847},
  year={2024},
  publisher={IEEE}
}

@article{wang2025joint,
  title={Joint Beamforming and Trajectory Design for Multi-Antenna {UAV-Enabled ISAC}},
  author={Wang, Weidan and Liu, Xin and Liu, Zechen},
  journal={IEEE Transactions on Vehicular Technology},
  year={2025},
  publisher={IEEE}
}

@article{mao2024uav,
  title={{UAV-assisted }communications in {SAGIN-ISAC}: Mobile user tracking and robust beamforming},
  author={Mao, Weihao and Lu, Yang and Pan, Gaofeng and Ai, Bo},
  journal={IEEE Journal on Selected Areas in Communications},
  year={2024},
  publisher={IEEE}
}

@article{benaya2025aerial,
  title={Aerial {ISAC}: A {HAPS}-assisted integrated sensing, communications and computing framework for enhanced coverage and security},
  author={Benaya, Ahmed M and Hassan, Mohamed S and Ismail, Mahmoud H and Landolsi, Taha},
  journal={IEEE Transactions on Green Communications and Networking},
  year={2025},
  publisher={IEEE}
}

@article{kirik2025isac,
  title={An {ISAC-Assisted} Beam Alignment Design for {HAP}-Based 6G Flying Ad-Hoc Networks},
  author={Kirik, Muhammet and Afeef, Liza and Arslan, H{\"U}Seyin},
  journal={IEEE Open Journal of the Communications Society},
  year={2025},
  publisher={IEEE}
}

@article{huang2025design,
  title={Design of frequency index modulated waveforms for integrated {SAR} and communication on high-altitude platforms {(HAPs)}},
  author={Huang, Bang and Ahmed, Sajid and Alouini, Mohamed-Slim},
  journal={IEEE Transactions on Communications},
  year={2025},
  publisher={IEEE}
}

@article{zhou2025energy,
  title={Energy-Efficient Target Area Imaging for {UAV-SAR-Based ISAC}: Beamforming Design and Trajectory Optimization},
  author={Zhou, Jiayi and Zhang, Xiangyin and Qin, Kaiyu and Yang, Feng and Wang, Libo and Song, Deyu},
  journal={Remote Sensing},
  volume={17},
  number={12},
  pages={2082},
  year={2025},
  publisher={MDPI}
}

@ARTICLE{Zhou2025JointTrajectory,
  author={Zhou, Jiayi and Zhang, Xiangyin and Qin, Kaiyu and Yang, Feng and Wang, Libo},
  journal={IEEE Access}, 
  title={Joint Trajectory and Power Optimization for {UAV-SAR Based ISAC} System}, 
  year={2025},
  volume={13},
  number={},
  pages={143465-143478},
  doi={10.1109/ACCESS.2025.3598359}}

@article{hu2022trajectory,
  title={Trajectory planning of cellular-connected {UAV} for communication-assisted radar sensing},
  author={Hu, Shuyan and Yuan, Xin and Ni, Wei and Wang, Xin},
  journal={IEEE Transactions on Communications},
  volume={70},
  number={9},
  pages={6385--6396},
  year={2022},
  publisher={IEEE}
}

@article{zhang2025design,
  title={Design of 3D Beamforming and Deployment Strategies for {ISAC-based HAPS} Systems},
  author={Zhang, Xue and Huang, Bang and Alouini, Mohamed-Slim},
  journal={arXiv preprint arXiv:2506.11294},
  year={2025}
}

@inproceedings{moro2024exploring,
  title={Exploring {ISAC} technology for {UAV SAR} imaging},
  author={Moro, Stefano and Linsalata, Francesco and Manzoni, Marco and Magarini, Maurizio and Tebaldini, Stefano},
  booktitle={ICC 2024-IEEE International Conference on Communications},
  pages={1582--1587},
  year={2024},
  organization={IEEE}
}

@article{kanani2025haps,
  title={{HAPS-ISAC}: Enhancing Sensing and Communication in 6G Networks with Advanced {MIMO} Beamforming},
  author={Kanani, Parisa and Omidi, Mohammad Javad and Modarres-Hashemi, Mahmoud and Yanikomeroglu, Halim},
  journal={IEEE Open Journal of the Communications Society},
  year={2025},
  publisher={IEEE}
}

@article{zhang2024joint,
  title={A joint {UAV} trajectory, user association, and beamforming design strategy for {multi-UAV assisted ISAC} systems},
  author={Zhang, Ruizhi and Zhang, Ying and Tang, Rui and Zhao, Huapeng and Xiao, Qing and Wang, Chenye},
  journal={IEEE Internet of Things Journal},
  year={2024},
  publisher={IEEE}
}

@article{bayessa2024joint,
  title={Joint {UAV} deployment and precoder optimization for multicasting and target sensing in {UAV-assisted ISAC} networks},
  author={Bayessa, Gezahegn Abdissa and Chai, Rong and Liang, Chengchao and Jain, Deepak Kumar and Chen, Qianbin},
  journal={IEEE Internet of Things Journal},
  year={2024},
  publisher={IEEE}
}

@ARTICLE{Yan2025uav,
  author={Yan, Shaoqiang and Luo, Hongliang and Yang, Ping and Zhao, Jianwei and Gao, Feifei},
  journal={IEEE Transactions on Wireless Communications}, 
  title={{UAV} Trajectory Monitoring for Integrated Sensing and Communications System}, 
  year={2025},
  volume={},
  number={},
  pages={1-1},
  doi={10.1109/TWC.2025.3598799}}

@article{jiang2025integrated,
  title={Integrated sensing and communication for low altitude economy: Opportunities and challenges},
  author={Jiang, Yihang and Li, Xiaoyang and Zhu, Guangxu and Li, Hang and Deng, Jing and Han, Kaifeng and Shen, Chao and Shi, Qingjiang and Zhang, Rui},
  journal={IEEE Communications Magazine},
  year={2025},
  publisher={IEEE}
}

@ARTICLE{Mozaffari2016UnmannedAerialVehicle,
  author={Mozaffari, Mohammad and Saad, Walid and Bennis, Mehdi and Debbah, Mérouane},
  journal={IEEE Transactions on Wireless Communications}, 
  title={Unmanned Aerial Vehicle With Underlaid Device-to-Device Communications: Performance and Tradeoffs}, 
  year={2016},
  volume={15},
  number={6},
  pages={3949-3963},
  doi={10.1109/TWC.2016.2531652}}

@ARTICLE{Lyu2017PlacementOptimization,
  author={Lyu, Jiangbin and Zeng, Yong and Zhang, Rui and Lim, Teng Joon},
  journal={IEEE Communications Letters}, 
  title={Placement Optimization of {UAV}-Mounted Mobile Base Stations}, 
  year={2017},
  volume={21},
  number={3},
  pages={604-607},
  doi={10.1109/LCOMM.2016.2633248}}

@article{wu2019fundamental,
  title={Fundamental trade-offs in communication and trajectory design for {UAV-enabled} wireless network},
  author={Wu, Qingqing and Liu, Liang and Zhang, Rui},
  journal={IEEE Wireless Communications},
  volume={26},
  number={1},
  pages={36--44},
  year={2019},
  publisher={IEEE}
}

@article{wu2018capacity,
  title={Capacity characterization of UAV-enabled two-user broadcast channel},
  author={Wu, Qingqing and Xu, Jie and Zhang, Rui},
  journal={IEEE Journal on Selected Areas in Communications},
  volume={36},
  number={9},
  pages={1955--1971},
  year={2018},
  publisher={IEEE}
}

@ARTICLE{DengYu2024TheHighResolutionSynthetic,
  author={Deng, Yunkai and Yu, Weidong and Wang, Pei and Xiao, Dengjun and Wang, Wei and Liu, Kaiyu and Zhang, Heng},
  journal={IEEE Geoscience and Remote Sensing Magazine}, 
  title={The High-Resolution Synthetic Aperture Radar System and Signal Processing Techniques: Current progress and future prospects [Review Papers]}, 
  year={2024},
  volume={12},
  number={4},
  pages={169-189},
  doi={10.1109/MGRS.2024.3456444}}

@ARTICLE{WuJiang2024AIEnhancedIntegrated,
  author={Wu, Nan and Jiang, Rongkun and Wang, Xinyi and Yang, Lyuxiao and Zhang, Kecheng and Yi, Wenqiang and Nallanathan, Arumugam},
  journal={IEEE Communications Magazine}, 
  title={{AI}-Enhanced Integrated Sensing and Communications: Advancements, Challenges, and Prospects}, 
  year={2024},
  volume={62},
  number={9},
  pages={144-150},
  doi={10.1109/MCOM.001.2300724}}

@ARTICLE{Zeng2016Wirelesscommunications,
  author={Zeng, Yong and Zhang, Rui and Lim, Teng Joon},
  journal={IEEE Communications Magazine}, 
  title={Wireless communications with unmanned aerial vehicles: opportunities and challenges}, 
  year={2016},
  volume={54},
  number={5},
  pages={36-42},
  doi={10.1109/MCOM.2016.7470933}}

@article{rago2025innovative,
  title={Innovative Multi-Layer Approaches for 6G Integrated Terrestrial And Non-Terrestrial Networks},
  author={Rago, Arcangela and Guidotti, Alessandro and Amatetti, Carla and Sambo, Nicola and Morosi, Simone and Sacchi, Claudio and Matera, Francesco and Settembre, Marina and Piro, Giuseppe and Coralli, Alessandro Vanelli and others},
  journal={IEEE Communications Standards Magazine},
  year={2025},
  publisher={IEEE}
}

@INPROCEEDINGS{JirousekPeichl2024DLR,
	author={Jirousek, Matthias and Peichl, Markus and Anger, Simon and Dill, Stephan and Limbach, Markus},
	booktitle={EUSAR 2024; 15th European Conference on Synthetic Aperture Radar}, 
	title={The {DLR} High Altitude Platform Synthetic Aperture Radar Instrument {HAPSAR}}, 
	year={2024},
	volume={},
	number={},
	pages={1244-1248},
address={Munich, Germany},
	doi={}}

@book{richards2005fundamentals,
  title={Fundamentals of radar signal processing},
  author={Richards, Mark A and others},
  volume={1},
  year={2005},
  publisher={Mcgraw-hill New York}
}

@article{Nabi2025JointOffloading,
  author={Nabi, Ahmadun and Moh, Sangman},
  journal={IEEE Transactions on Mobile Computing}, 
  title={Joint Offloading Decision, User Association, and Resource Allocation in Hierarchical Aerial Computing: Collaboration of {UAVs} and {HAP}}, 
  year={2025},
  volume={},
  number={},
  pages={1-17},
  doi={10.1109/TMC.2025.3548668}
}

@article{Wu2025MultiHAP,
 author={Wu, Wenhao and Feng, Wei and Fang, Yi and Lin, Zhijian and Lu, Xiaoqiang},
  title   = {Multi-{HAP}-Assisted Computation Offloading in Space–Air–Ground–Sea Integrated Network},
  journal = {IEEE Internet of Things Journal},
  year    = {2025},
  volume={},
  number={},
  pages={1-1},
  doi     = {10.1109/JIOT.2025.3548156}
}

@ARTICLE{Yu2024MECNOMA,
  author={Yu, Xiangbin and Zhang, Xinyi and Rui, Yun and Wang, Kezhi and Dang, Xiaoyu and Guizani, Mohsen},
  journal={IEEE Journal on Selected Areas in Communications}, 
  title={Joint Resource Allocations for Energy Consumption Optimization in {HAPS-Aided MEC-NOMA} Systems}, 
  year={2024},
  volume={42},
  number={12},
  pages={3632-3646},
  doi={10.1109/JSAC.2024.3459084}}

@article{Ali2024PowerLEO,
    author={Ali, Zain and Rezki, Zouheir and Alouini, Mohamed-Slim},
  title   = {Optimizing Power Allocation in {HAPS}-Assisted {LEO} Satellite Communications},
  journal = {IEEE Transactions on Machine Learning in Communications and Networking},
year={2024},
  volume={2},
  number={},
  pages={1661-1677},
  doi={10.1109/TMLCN.2024.3491054}}

@article{abbasi2024haps,
  title={{HAPS} for {6G} networks: Potential use cases, open challenges, and possible solutions},
  author={Abbasi, Omid and Yadav, Animesh and Yanikomeroglu, Halim and {\DJ}{\`a}o, Ngọc-D{\~u}ng and Senarath, Gamini and Zhu, Peiying},
  journal={IEEE Wireless Communications},
  volume={31},
  number={3},
  pages={324--331},
month={June},
  year={2024},
  publisher={IEEE}
}

@ARTICLE{Javed2025SystemDesign,
  author={Javed, Sidrah and Alouini, Mohamed-Slim},
  journal={IEEE Transactions on Wireless Communications}, 
  title={System Design and Parameter Optimization for Remote Coverage From {NOMA}-Based High-Altitude Platform Stations {(HAPS)}}, 
  year={2025},
  volume={24},
  number={2},
  pages={1387-1400},
  doi={10.1109/TWC.2024.3508872}}

@inproceedings{lahmeri2022trajectory,
  title={Trajectory and resource optimization for UAV synthetic aperture radar},
  author={Lahmeri, Mohamed-Amine and Ghanem, Walid and Knill, Christina and Schober, Robert},
  booktitle={2022 IEEE Globecom Workshops (GC Wkshps)},
  pages={897--903},
  year={2022},
  organization={IEEE}
}

@misc{grant2014cvx,
  title={{CVX}: {Matlab} software for disciplined convex programming, version 2.1},
  author={Grant, Michael and Boyd, Stephen},
  year={2014}
}

@article{javed2023interdisciplinary,
  title={An Interdisciplinary Approach to Optimal Communication and Flight Operation of High-Altitude Long-Endurance Platforms},
  author={Javed, Sidrah and Alouini, Mohamed-Slim and Ding, Zhiguo},
  journal={IEEE Transactions on Aerospace and Electronic Systems},
  year={2023},
  publisher={IEEE}
}
\end{document}